\documentclass[numbook,envcountsect,envcountsame,envcountreset,runningheads,smallextended]{svjour3}
\smartqed  

\usepackage{makeidx}
\usepackage{graphicx}
\usepackage{amsmath}
\usepackage{amsfonts}
\usepackage{enumerate}
\usepackage{hyperref}
\usepackage{caption}
 \usepackage{mathptmx}      
\usepackage{epstopdf}

\hypersetup{colorlinks=true,dvips}

\newcommand{\fracs}[2]{{ \textstyle \frac{#1}{#2} }}

\def \eps {{\epsilon}}

\def \D {{\rm d}}

\def \hP {{\widehat{P}}}

\def \hS {{\widehat{S}}}

\def \hY {{\widehat{Y}}}

\def \VV {{\mathbb{V}}}
\def \EE {{\mathbb{E}}}
\def \ii {{\bf 1}}

\DeclareMathOperator{\bO}{{             {\cal O}}}
\DeclareMathOperator{\bo}{{\scriptscriptstyle {\cal O}}}

\def \D {{\rm d}}

\def \eps {{\epsilon}}

\def \P {{\mathbb{P}}}

\def\JELname{{\bfseries JEL Classification}\enspace}
      \def\JEL#1{\par\addvspace\medskipamount{\rightskip=0pt plus1cm
      \def\and{\ifhmode\unskip\nobreak\fi\ $\cdot$
      }\noindent\JELname\ignorespaces#1\par}}
\begin{document}

\title{Multilevel Monte Carlo For Exponential L\'{e}vy Models}

\author{Michael B. Giles\and Yuan Xia}


\institute{Mike Giles \at
              Mathematical Institute and Oxford-Man Institute of Quantitative Finance,
Oxford University \\
\email{mike.giles@maths.ox.ac.uk}
           \and
           Yuan Xia\at
              Mathematical Institute and Oxford-Man Institute of Quantitative Finance,
Oxford University\\
\email{yuan.xia.cn@gmail.com}  
}
\date{Received: date / Accepted: date}

\maketitle


\begin{abstract}
We apply the multilevel Monte Carlo method for option pricing problems 
using exponential L\'{e}vy models with a uniform timestep discretisation.
For lookback and barrier options, we derive estimates of the convergence 
rate of the error introduced by the discrete monitoring of the running 
supremum of a broad class of L\'{e}vy  processes. We then use these to
obtain upper bounds on the multilevel Monte Carlo variance convergence rate
for the Variance Gamma, NIG and $\alpha$-stable processes. 
We also provide analysis of a trapezoidal approximation for Asian options.
Our method is illustrated by numerical experiments. 

\JEL{C15\and C63}
\end{abstract}

\section{Introduction}

Exponential L\'{e}vy models are based on the assumption that asset returns
follow a L\'{e}vy process \cite{schoutens03,ct04}. The asset price follows
\begin{equation}
S_{t}=S_{0}\exp \left( X_{t}\right)  \label{exp-levy}
\end{equation}%
where $X$ is an $(m,\sigma ,\nu )$-L\'{e}vy process
\begin{equation*}
X_{t}=mt+\sigma B_{t}+\int_{0}^{t}\int_{\{|z|\ge 1\}}z\ J(%
\D z,\D s)+\int_{0}^{t}\int_{\{|z|< 1\}}z\left( J(\D z,\D %
s)-\nu (\D z)\D s\right)
\end{equation*}%
where $m$ is a constant, $B_{t}$ is a Brownian Motion, $\ J$ is the jump
measure and $\nu $ is the L\'{e}vy measure(c.f. Theorem 42 in \cite{protter04}).

Models with jumps give an intuitive explanation of implied volatility skew
and smile in the index option market and foreign exchange market(Chapter 11 in \cite{ct04})%
. The jump fear is mainly on the downside in the equity market which
produces a premium for low-strike options; the jump risk is symmetric in the
foreign exchange market so the implied volatility has a smile shape. Chapter 7 in \cite%
{ct04} shows that models building on pure jump processes can reproduce the
stylized facts of asset returns, like heavy tails and the asymmetric
distribution of increments. Since pure jump processes of finite activity
without a diffusion component cannot generate a realistic path, it is
natural to allow the jump activity to be infinite. In this work we deal with
infinite-activity pure jump exponential L\'{e}vy models, in particular
models driven by Variance Gamma (VG), Normal Inverse Gaussian (NIG) and $%
\alpha $-stable processes which allow direct simulation of increments.

We are interested in estimating the expected payoff value $\mathbb{E}[f(S)]$
in option pricing problems. In the case of European options, it is possible
to directly sample the final value of the underlying L\'{e}vy process,
but in the case of Asian, lookback and barrier options the option value
depends on functionals of the L\'{e}vy process and so it is necessary to
approximate those. In the case of a VG model with a lookback option, the
convergence results in \cite{dl11} show that to achieve an $\bO(\eps)$ root
mean square (RMS) error using a standard Monte Carlo method with a uniform
timestep discretisation requires $\bO(\eps^{-2})$ paths, each with 
$\bO(\eps^{-1})$ timesteps, leading to a computational complexity of 
$\bO(\eps^{-3})$.

In the case of simple Brownian diffusion, Giles \cite%
{giles07b,giles08b} introduced a multilevel Monte Carlo (MLMC) method,
reducing the computational complexity from $\bO(\eps^{-3})$ to $\bO(\eps^{-2})$
for a variety of payoffs. The objective of this paper is to investigate
whether similar benefits can be obtained for exponential L\'{e}vy processes.

Various researchers have investigated simulation methods for the running
maximum of L\'{e}vy processes. Reference \cite{ft12} develops an adaptive Monte Carlo
method for functionals of killed L\'{e}vy processes with a controlled bias.
Small-time asymptotic expansions of the exit probability are given with
computable error bounds. For evaluating the exit probability when the
barrier is close to the starting point of the process, this algorithm
outperforms a uniform discretisation significantly. Reference \cite{kkpv12} develops a
novel Wiener-Hopf Monte-Carlo method to generate the joint distribution of $%
\left(X_{T},\sup_{0\leq t\leq T}X_{t}\right) $ which is further extended to
MLMC in \cite{fkss12}, obtaining an RMS error $\eps$ with a computational
complexity of $\bO\left( \eps^{-3}\right) $ for  L\'{e}vy processes with bounded variation and $\bO\left( \eps^{-4}\right) $ for processes with infinite variation. The method
currently cannot be directly applied to VG, NIG and $\alpha-$stable processes.  
References \cite{dh10,dereich11} adapt MLMC to L\'{e}vy-driven SDEs with payoffs which are
Lipschitz w.r.t.~the supremum norm. If the L\'{e}vy process does not
incorporate a Brownian process, reference \cite{dereich11} obtains an $\bO\left( \eps%
^{-(6\beta)/(4-\beta )}\right)$ upper bound on the worst case computational
complexity, where $\beta$ is the BG index which will be defined later.

In contrast to those advanced techniques, we take the discretely monitored
maximum based on a uniform timestep discretisation of the L\'{e}vy process as the
approximation.
The outline of the work is as follows. First we review the Multilevel Monte
Carlo method and present the three L{\'e}vy processes we will consider in
our numerical experiments.
To prepare for the analysis of the multilevel
variance of lookback and barrier, we bound the convergence rate of the
discretely monitored running maximum for a large class of L\'{e}vy processes
whose L\'{e}vy measures have a power law behavior for small jumps, and have
exponential tails. Based on this, we conclude by bounding the variance of
the multilevel estimators. Numerical results are then presented for the 
multilevel Monte Carlo applied to Asian, lookback and barrier options using 
the three different exponential L\'{e}vy models.

\section{Multilevel Monte Carlo (MLMC) method}

For a path-dependent payoff $P$ based on an exponential L\'{e}vy model on
the time interval $[0,T]$, let $\hP_{\ell }$ denote its approximation using a
discretisation with $M^{\ell }$ uniform timesteps of size $h_{\ell
}=M^{-\ell }\,T$ on level $\ell$; in the numerical results reported later,
we use $M\!=\!2$. 
Due to the linearity of the expectation
operator, we have the following identity:
\begin{equation}
\EE[\hP_L]=\EE[\hP_0]+\sum_{\ell =1}^{L}\EE[\hP_\ell \!-\! \hP_{\ell-1}].
\label{eq:identity}
\end{equation}%
Let $\hY_{0}$ denote the standard Monte Carlo estimate for $\EE[\hP_0]$
using $N_{0}$ paths, and for $\ell >0$, we use $N_{\ell }$ independent paths
to estimate $\EE[\hP_\ell \!-\! \hP_{\ell-1}]$ using
\begin{equation}
\hY_{\ell }=N_{\ell }^{-1}\sum_{i=1}^{N_{\ell }}\left( \hP_{\ell }^{(i)}\!-\!%
\hP_{\ell -1}^{(i)}\right) .  \label{eq:est_l}
\end{equation}%
For a given path generated for $\hP_{\ell }^{(i)}$, we can calculate $\hP%
_{\ell -1}^{(i)}$ using the same underlying L\'{e}vy path. The multilevel
method exploits the fact that $V_{\ell }:=\VV[\hP_\ell \!-\!\hP_{\ell-1}]$
decreases with $\ell $, and adaptively chooses $N_{\ell }$ to minimise the
computational cost to achieve a desired RMS error. This is summarized in the
following theorem in \cite{giles12,giles15}:

\begin{theorem}
\label{thm:cc} Let $P$ denote a functional of $S_{t}$, and let $\hP_{\ell }$
denote the corresponding approximation using a discretisation with uniform
timestep $h_{\ell }=M^{-\ell }\,T$. If there exist independent estimators $%
\hY_{\ell }$ based on $N_{\ell }$ Monte Carlo samples, each with complexity $C_\ell$,
and positive
constants $\alpha,\beta ,c_{1},c_{2},c_{3}$ such that $\alpha\!\geq \!{\ %
\textstyle\frac{1}{2}}\min(1,\beta)$ and

\begin{itemize}
\item[i) ] $\displaystyle\left\vert \EE[\hP_\ell - P]\right\vert \leq
c_{1}\,h_{\ell }^{\alpha }$

\item[ii) ] $\displaystyle\EE[\hY_\ell]=\left\{
\begin{array}{ll}
\EE[\hP_0], & \ell =0 \\[0.1in]
\EE[\hP_\ell - \hP_{\ell-1}], & \ell >0%
\end{array}%
\right. $

\item[iii) ] $\displaystyle\ \VV[\hY_\ell]\leq c_{2}\,N_{\ell }^{-1}h_{\ell
}^{\beta }$

\item[iv) ] $\displaystyle\ C_{\ell }\leq c_{3}\,N_{\ell }\,h_{\ell }^{-1}$,
\end{itemize}

\noindent then there exists a positive constant $c_{4}$ such that for any $%
\eps\!<\!e^{-1}$ there are values $L$ and $N_{\ell }$ for which the
multilevel estimator
\begin{equation*}
\hY=\sum_{\ell =0}^{L}\hY_{\ell },
\end{equation*}%
has a mean-square-error with bound
\begin{equation*}
MSE\equiv \EE\left[ ( \hY-\EE[P]) ^{2}\right] <\eps^{2}
\end{equation*}%
with a computational complexity $C$ with bound
\begin{equation*}
C\leq \left\{
\begin{array}{ll}
c_{4}\,\eps^{-2}, & \ \beta >1, \\[0.1in]
c_{4}\,\eps^{-2}(\log \eps)^{2}, & \ \beta =1, \\[0.1in]
c_{4}\,\eps^{-2-(1-\beta )/\alpha }, & \ 0<\beta <1.%
\end{array}%
\right.
\end{equation*}
\end{theorem}

We will focus on the multilevel variance convergence rate $\beta$ in the following numerical results and analysis since it is crucial in determining the computational complexity.


\section{L{\'e}vy models}

The numerical results to be presented later use the following three
models.

\subsection{Variance Gamma (VG) \label{sec:VG}}

The VG process with parameter set $(\sigma ,\theta ,\kappa )$ is the L\'{e}%
vy process $X$ with characteristic function $\mathbb{E[}\exp
(iuX_{t})]=(1-iu\theta \kappa +{\ \textstyle \frac{1}{2} }\sigma ^{2}
u^{2}\kappa)^{-t/\kappa }. $ The L\'{e}vy measure of the VG process is (\cite%
Table 4.5 in {ct04})
\begin{equation*}
\nu (x)=\frac{1}{\kappa \left\vert x\right\vert }e^{A-B\left\vert
x\right\vert }\ \ \text{with\ \ }A=\frac{\theta }{\sigma ^{2}}\ \ \text{and\
\ }B=\frac{\sqrt{\theta ^{2}+2\sigma ^{2}/\kappa }}{\sigma ^{2}}.
\end{equation*}
One advantage of the VG process is that its additional parameters
make it possible to fit the skewness and kurtosis of the stock 
returns (section 7.3 in \cite{ct04}).
%
%
Another is that it is easily simulated as we have a subordinator
representation $X_{t}=\theta G_{t}+\sigma B_{G_{t}}$ in which $B$ is a
Brownian process and the subordinator $G$ is a Gamma process with
parameters ($1/\kappa,1/\kappa $).

For the ease of computation, we follow the mean-correcting pricing measure
in section 6.2.2 in \cite{schoutens03}, with risk-free interest rate $r=0.05$. 
Let $\exp(-rt)S_t$ be a martingale. This results in the drift being
\[
m = r + \kappa^{-1} \log(1 + \theta \kappa - \fracs{1}{2} \sigma^2 \kappa).
\]
After transforming the parameter representation to the definition we use, the calibration in table 6.3 in \cite{schoutens03} gives 
$\sigma =0.1213,\theta =-0.1436,\kappa=0.1686$.

\subsection{Normal Inverse Gaussian (NIG) \label{sec:NIG}}

The NIG process with parameter set $(\sigma,\theta ,\kappa )$ is the L\'{e}%
vy process $X$ with characteristic function $\mathbb{E[}\exp
(iuX_{t})]=\exp \left( \frac{t}{\kappa }-\frac{t}{\kappa }\sqrt{1-2iu\theta
\kappa +\kappa \sigma ^{2}u^{2}}\right) $ and L\'{e}vy measure
\begin{equation*}
\nu (x)=\frac{C}{\kappa \left\vert x\right\vert }e^{Ax}K_{1}\left(
B\left\vert x\right\vert \right) \ \text{with\ }A=\frac{\theta }{\sigma ^{2}}%
,\ B=\frac{\sqrt{\theta ^{2}+\sigma ^{2}/\kappa }}{\sigma ^{2}},\text{\ }C=%
\frac{\sqrt{\theta ^{2}+2\sigma ^{2}/\kappa }}{2\pi \sigma \sqrt{\kappa }}.
\end{equation*}%
$K_{n}\left( x\right) $ is the modified Bessel function of the second kind (%
section 4.4.3 in \cite{ct04}). As $x\rightarrow 0$, $K_{1}\left( x\right) \sim \frac{1}{x%
}+\bO\left( 1\right), $ 
while as $x\rightarrow \infty$, $K_{1}\left( x\right) \sim e^{-x}\sqrt{\frac{%
\pi }{2\left\vert x\right\vert }}\left( 1+\bO\left( \frac{1}{\left\vert
x\right\vert }\right) \right) . $

In terms of simulation, the NIG process can be represented as $X_{t}=\theta
I_{t}+\sigma B_{I_{t}}$, where the subordinator $I_{t}$ is an Inverse
Gaussian process with parameters $(\frac{1}{\kappa}, 1)$. 
Algorithm 6.9 in \cite{ct04} can be used to generate Inverse Gaussian samples.

Using the mean-correcting pricing measure leads to
\[
m = r - \kappa^{-1} + \pi C B \kappa^{-1}\sqrt{B^2 - (A+1)^2}.
\]
Following the calibration in \cite{schoutens03} we use the parameters 
$\sigma=0.1836,\theta =-0.1313,\kappa =1.2819$, and again use risk-free 
interest rate $r=0.05$.

\subsection{ Spectrally negative $\protect\alpha $-stable process}

The scalar spectrally negative $\alpha $-stable process has a L\'{e}vy measure of the form; see section 1.2.6 in \cite{kyprianou06}:
\begin{equation*}
\nu (x)=\frac{B}{\left\vert
x\right\vert ^{\alpha +1}}1_{\left\{ x<0\right\} }
\end{equation*}%
for $0<\alpha <2$ and some non-negative $B$. We follow the reference
to discuss another parameterisation of $\alpha $-stable process with
characteristic function%
\begin{equation}
\begin{array}{l}
\mathbb{E[}\exp (iuX_{t})]=\exp \left\{ -\!t\!B^{\alpha
}\left\vert u\right\vert ^{\alpha }\left( 1 + i\text{sgn}\left(
u\right) \tan \frac{\pi \alpha }{2}\right) \right\} ,\text{ if }\alpha \neq
1, \\
\mathbb{E[}\exp (iuX_{t})]=\exp \left\{ -\!t\!B \left\vert
u\right\vert \left( 1+i\frac{2}{\pi }\text{sgn}\left(
u\right) \log \left\vert u\right\vert \right) \right\} ,\text{ if }\alpha =1,%
\end{array}
\label{def:ch_fun_a-stable}
\end{equation}%
where sgn$\left( u\right) \!=\!\left\vert u\right\vert /u$ if $u\!\neq\! 0$ 
and sgn$\left(0\right)\!=\!0$.There 
are no positive jumps for the spectrally negative  process, which has a finite 
exponential moment $\mathbb{E[}\exp (uX_{t})]$ \cite{cw03}.

For this case, the mean-correcting drift is
\[
m = r +B^\alpha \sec\fracs{\alpha \pi}{2}.
\]
Sample paths of $\alpha $-stable processes can be generated by the algorithm
in \cite{cms76}. 
Following \cite{cw03}, we use the parameters $\alpha\!=\!1.5597$ and $B\!=\!0.1486$.

\section{Key numerical analysis results}

The variance of the multilevel correction, $V_{\ell }=\VV[\hP_{\ell}-\hP_{\ell-1}]$
depends on the behavior of the difference between the continuously and
discretely monitored suprema of $X_t$, defined for a unit time interval as
\begin{equation*}
D_n = \sup_{0\leq t\leq 1}X_t\ -\max_{i=0,1,\ldots ,n}X_{i/n}.
\end{equation*}

To derive the order of weak convergence for lookback-type payoffs, we are
concerned with $\mathbb{E}\left[ D_{n}\right]$, which is extensively studied
in the literature. For example, \cite{dl11}, \cite{chen11} and \cite{cfs11}
derive asymptotic expansions for jump-diffusion, VG, NIG processes, as well
as estimates for general L\'{e}vy processes, by using Spitzer's identity
\cite{spitzer56}.

A key result due to Chen \cite{chen11} is the following:

\begin{theorem}
\label{thm:mean}Suppose $X$ is a scalar L\'{e}vy process with triple $%
(m,\sigma ,\nu ),$ with finite first moment, i.e.%
\begin{equation*}
\int_{\{ |x| >1\}}\left\vert x\right\vert \ \nu (\D x)<\infty
.
\end{equation*}%
Then $\displaystyle D_{n}=\sup_{0\leq t\leq 1}X_{t}-\max_{i=0,1,\ldots
,n}X_{i/n}\ $ satisfies

\begin{enumerate}
\item If $\sigma >0$%
\begin{equation*}
\mathbb{E}\left[ D_{n}\right] =\bO(1/ \sqrt{n}) ;
\end{equation*}

\item If $\sigma =0$ and $X$ is of finite variation, i.e. $%
\int_{\{ |x| <1\}} |x| \nu (\D x)<\infty $%
\begin{equation*}
\mathbb{E}\left[ D_{n}\right] =\bO( \log n /n) ;
\end{equation*}

\item If $\sigma =0$ and $X$ is of infinite variation, then
\begin{equation*}
\mathbb{E}\left[ D_{n}\right] =\bo( n^{-1/\beta + \delta}) ,
\end{equation*}%
where
\begin{equation*}
\beta =\inf \left\{ \alpha >0:\int_{\{ |x| <1\}}\left\vert
x\right\vert ^{\alpha }\nu (\D x)<\infty \right\}
\end{equation*}%
is the Blumenthal-Getoor index of $X$, and $\delta \!>\!0$ is an
arbitrarily small strictly positive constant.
\end{enumerate}
\end{theorem}

The VG process has finite variation with Blumenthal-Getoor index $0$; the
NIG process has infinite variation with Blumenthal-Getoor index $1$. They
correspond to the second and third cases of Theorem \ref{thm:mean}
respectively.


For the multilevel variance analysis we require higher moments of $D_n$. In
the pure Brownian case, Asmussen {\it et al} (\cite{agp1995}) obtain the
asymptotic distribution of $D_{n}$, which in turn gives the asymptotic
behavior of $\mathbb{E}[ D_{n}^{2}]$. \cite{dl11} extends the
result to finite activity jump processes with non-zero diffusion.

However, in this paper we are looking at infinite activity jump processes.
Our main new result is therefore concerned with the $L^{p}$ convergence rate
of $D_{n}$ for pure jump L\'{e}vy processes. This will be used later to
bound the variance of the Multilevel Monte Carlo correction term $V_{\ell }$
for both lookback and barrier options.

\begin{theorem}
\label{prop:Un}Let $X$ be a scalar pure jump L\'{e}vy process, and
suppose its L\'{e}vy measure $\nu (x)$ satisfies
\begin{equation}
\begin{array}{l}
C_{2}\left\vert x\right\vert ^{-1-\alpha }\leq \nu (x)\leq C_{1}\left\vert
x\right\vert ^{-1-\alpha },\ \text{for }\left\vert x\right\vert \leq 1; \\%
[0.1in]
\ \nu (x)\leq \exp \left( -C_{3}\left\vert x\right\vert \right) ,\ \text{for
}\left\vert x\right\vert >1,%
\end{array}
\label{ass:prop:Un}
\end{equation}%
where $C_{1},C_{2},C_{3}>0,\ 0\leq \alpha <2$ are constants. Then for $p\geq
1$
\begin{equation*}
D_{n}=\sup_{0\leq t\leq 1}X_{t}\ -\ \max_{i=0,1,\ldots ,n}X_{i/n} \ \
\end{equation*}
satisfies
\begin{equation*}
\mathbb{E}\left[ D_{n}^{p}\right] =\left\{
\begin{array}{ll}
\bO\left( 1 / n\right), & p>2\alpha ; \\[0.1in]
\bO\left( \left( \log n / n\right) ^{\frac{p}{2\alpha}}\right), &
p\leq 2\alpha .%
\end{array}%
\right.
\end{equation*}

\vspace{0.2in}

If, in addition, $X_{t}$ is spectrally negative, i.e.~$\nu(x)=0$ for $x>0$,
then

\begin{equation*}
\mathbb{E}\left[ D_n^p\right] =\left\{
\begin{array}{ll}
\bO\left( n^{-p}\right), & 0\leq \alpha <1; \\[0.1in]
\bo\left( n^{-p/\alpha + \delta}\right), & 1\leq \alpha <2;%
\end{array}%
\right.
\end{equation*}
for any $\delta>0$.
\end{theorem}

\vspace{0.2in}

We will give the proof of this result later in Section \ref{sec:proofs-4.2}.
Note that for $p\!=\!1$, the general bound  in Theorem \ref%
{prop:Un} is  slightly sharper than Chen's result  for $\alpha\!<\!\frac{1}{2}$, is the same for $\alpha\!=\!\frac{1}{2}$, and is not as tight as Chen's result for $\frac{1}{2}\!<\! \alpha\!<\!2$; the spectrally negative bound is slightly sharper than Chen's result for $\alpha\!<\!1$, and the
bound is the same for $1\!\leq\! \alpha\!<\!2$. 


\section{MLMC analysis}

\subsection{Asian options}

We consider the analysis for a Lipschitz arithmetic Asian payoff 
$P=P( \overline{S})$ where 
\begin{equation*}
\overline{S}=S_{0}\ T^{-1}\int_{0}^{T}\exp \left( X_{t}\right) \D t\ .
\end{equation*}
and $P$ is Lipschitz such that $\vert P(S_{1})-P(S_{2})\vert
\leq L_{K}\vert S_{1}-S_{2}\vert $.

\bigskip We approximate the integral using a trapezoidal approximation:%
\begin{equation}
\overline{\hS}:=S_{0}\ T^{-1} \sum_{j=0}^{n-1}{\textstyle\frac{1}{2}}\,h\,\,\left(
\exp \left( X_{jh}\right) \!+\!\exp \left( X_{\left( j+1\right) h}\right)
\right) ,  \label{def:appro_ar_asian}
\end{equation}%
and the approximated payoff is then $\hP=P( \overline{\hS})$.

\begin{proposition}
\label{prop:arith_asian}Let $X$ be a scalar L\'{e}vy process underlying
an exponential L\'{e}vy model. If $\overline{\hS}, \overline{S}$ are as defined 
above, and $\int_{\{ |z| >1\}}e^{2z}\ \nu (\D z)<\infty $, then
\begin{equation*}
\mathbb{E}\left[ ( \overline{\hS}-\overline{S}) ^{2}\right] =\bO%
(h^{2}).
\end{equation*}
\end{proposition}

The proof will be given later in Section \ref{sec:proofs-5.1}.  Using the Lipschitz property, the weak 
convergence for the numerical approximation is given by
\begin{equation*}
\left\vert \mathbb{E}[ \hP_\ell-P] \right\vert \leq
L_{K} \mathbb{E}\left [\vert \overline{\hS_\ell}-\overline{S} \vert \right]
 \leq L_{K}\left( \mathbb{E}[ ( \overline{\hS}-
\overline{S}) ^{2}] \right)^{1/2},
\end{equation*}
while the convergence of the MLMC variance follows from
\begin{eqnarray*}
V_{\ell }\ &\leq &\mathbb{E}\left[ ( \hP_{\ell}-\hP_{\ell-1}) ^{2}%
\right] \\
&\leq &2\ \mathbb{E}\left[ ( \hP_{\ell}-P)^{2}\right] +2\ \mathbb{E}%
\left[ ( \hP_{\ell-1}-P) ^{2}\right] \\
&\leq &2L_{K}^{2}\ \mathbb{E}\left[ ( \overline{\hS_{\ell}}-\overline{S}
) ^{2}\right] +2L_{K}^{2}\ \mathbb{E}\left[ ( \overline{\hS}_{\ell-1}-
\overline{S}) ^{2}\right] .
\end{eqnarray*}

\subsection{Lookback options}

\label{sec:lookback}

In exponential L\'{e}vy models, the moment generating function $\mathbb{E}%
\left[ \exp \left( q\sup_{0\leq t\leq T}X_{t}\right) \right]$ can be
infinite for large value of $q$. To avoid problems due to this, we consider
a lookback put option which has a bounded payoff
\begin{equation}
P = \exp(-r T) \, \left(K-S_0\exp(m)\right)^+,  \label{def:lookback_put}
\end{equation}
where $m = \sup_{0\leq t\leq T} X_t. $ Note that $P$ is a Lipschitz function
of $m$, since we have $\left\vert P^{\prime}(x) \right\vert \leq K$. Without
loss of generality, we assume $T=1$ in the following.

Because of the Lipschitz property, we have $\left| \mathbb{E}[ P\!-\!\hP%
_\ell] \right| \leq K\, \mathbb{E}[D_n] $ where $n\!=\!M^\ell\!=\!h_%
\ell^{-1} $. Therefore we obtain weak convergence for the processes covered
by Theorem \ref{thm:mean}, with the convergence rate given by the Theorem.

To analyse the variance, $V_{\ell }=\mathbb{V}[ \hP_{\ell }\!-\!\hP%
_{\ell-1}] $, we first note that
\begin{equation*};
0 \leq \max_{0\leq i \leq M^\ell} X_{i/M^\ell} - \max_{0\leq i \leq
M^{\ell-1}} X_{i/M^{\ell-1}} \leq \sup_{0\leq t \leq 1} X_t - \max_{0\leq i
\leq M^{\ell-1}} X_{i/M^{\ell-1}} = D_n
\end{equation*}
where $n=M^{\ell -1}$. Hence, we have
\begin{equation*}
V_{\ell } \ \leq\ \mathbb{E}\left[ ( \hP_{\ell }-\hP_{\ell -1})
^{2}\right] \ \leq\ K^{2}\ \mathbb{E}[D_{n}^{2}].
\end{equation*}
Theorem \ref{prop:Un} then provides the following bounds on the variance for
the VG, NIG and spectrally negative $\alpha $-stable processes.

\begin{proposition}
\label{prop:lookback} Let $X$ be a scalar L\'{e}vy process underlying an
exponential L\'{e}vy model. For the Lipschitz lookback put payoff (\ref%
{def:lookback_put}), we have the following multilevel variance convergence
rate results:

\begin{enumerate}
\item If $X$ is a Variance Gamma (VG) process, then $V_\ell =\bO\left(
h_{\ell }\right)$;

\item If $X$ is a Normal Inverse Gaussian (NIG) process, then $V_\ell = %
\bO\left( h_{\ell }\left\vert \log h_{\ell }\right\vert \right)$;

\item If $X$ is a spectrally negative $\alpha $-stable process with $%
\alpha >1$, then $V_\ell = \bo\left( h_{\ell }^{2/\alpha -\delta }\right)$,
for any small $\delta >0.$
\end{enumerate}
\end{proposition}


\subsection{Barrier options}

We consider a bounded up-and-out barrier option with discounted payoff
\begin{equation}
P\ =\ \exp (-rT)\,f(S_{T}\!)\ \mathbf{1}_{\left\{ \sup_{0<t<T}S_t<B\right\} }
 \ =\ \exp (-rT)\,f(S_{T}\!)\ \mathbf{1}_{\left\{ m<\log(B/S_0)\right\} },
\label{def:barrier}
\end{equation}%
where again $m=\sup_{0<t<T}X_t$, and $\left| f\left( x\right) \right|
\!\leq\! F$ is bounded. On level $\ell$, the numerical approximation is
\begin{equation}
\hP_\ell = \exp (-rT)\,\,f(S_{T}\!)\ \mathbf{1}_{\left\{ \widehat{m}_\ell<\log(B/S_0)\right\} }.  \label{def:appr_barr}
\end{equation}
where $\widehat{m}_\ell=\max_{0\leq i \leq M^{\ell}}X_{i h_\ell} $.

Our analysis for NIG and the spectrally negative $\alpha$-stable processes
requires the following quite general result.

\begin{proposition}
\label{prop:avikainen} If $m$ is a random variable with a locally bounded density
in a neighbourhood of $B$, and $\widehat{m}$ is a numerical approximation to $m$, then for any $%
p>0$ there exists a constant $C_p(B)$ such that
\begin{equation*}
\mathbb{E}\left[\, \left|\mathbf{1}_{\{ m<B\}} - \mathbf{1}_{\{ \widehat{m}%
<B\}} \right|\, \right] < C_p(B) \ \| m - \widehat{m} \|_p^{p/(p+1)}.
\end{equation*}
\end{proposition}

\begin{proof}
This  result was first proved by Avikainen (Lemma 3.4 in \cite{avikainen09}),
but we give here a simpler proof. If, for some fixed $X\!>\!0$, we have $|m
\!-\! B| > X$ and $|m \!-\! \widehat{m}| < X$, then $\mathbf{1}_{m<B} -
\mathbf{1}_{\widehat{m}<B} = 0$. Hence,
\begin{eqnarray*}
\mathbb{E}[\, \left| \mathbf{1}_{m<B} - \mathbf{1}_{\widehat{m}<B}
\right|\,] & \leq & \P[|m \!-\! B| \leq X] + \P[|m \!-\!
\widehat{m}| \geq X] \\
& \leq & 2\, \rho_{sup}(B)\, X + X^{-p} \ \| m - \widehat{m} \|^p_p
\end{eqnarray*}
with the first term being due to the local bound $\rho_{sup}(B)$ of $m$'s density and the second term due
to the Markov inequality. Differentiating the upper bound w.r.t.~$X$, we
find that it is minimised by choosing $X^{p+1} = \frac{p}{2\,\rho_{sup}(B)} \
\| m - \widehat{m} \|^p_p, $ and we then get the desired bound.
\end{proof}

\bigskip

Our analysis for the Variance Gamma process requires a sharper result
customised to the properties of L\'{e}vy processes.

\begin{proposition}
\label{prop:VG-barrier} If $X_t$ is a scalar pure jump L\'{e}vy process
satisfying the conditions of Theorem \ref{prop:Un} with $0\leq \alpha \leq{\ %
\textstyle \frac{1}{2} }$,  and $m$ and $\widehat{m}_n$ are the continuously
and discretely monitored suprema of $X_t$ and $m$ has a locally bounded 
density in a neighbourhood of $B$, then
\begin{equation*}
\mathbb{E}\left[\, \left|\mathbf{1}_{\{ m<B\}} - \mathbf{1}_{\{ \widehat{m}%
<B\}} \right|\, \right] = \bo(n^{-1/(1+2\alpha) + \delta}),
\end{equation*}
for any $\delta>0$.
\end{proposition}

The proof is given later in Section \ref{sec:proofs-5.4}.

\bigskip
Both of the above propositions require the condition that the supremum 
$m$ has a locally bounded density for all strictly positive values.
There is considerable current research on the supremum of 
L\'{e}vy processes \cite{chaumont13,cm15,kuznetsov11,kmr13}.
In particular, the comments following Proposition 2 in \cite{cm15}
indicate that the condition is satisfied by stable processes, and by 
a wide class of symmetric subordinated Brownian motions.  
Unfortunately, the VG and NIG processes in the current paper are 
not symmetric, so at present they lie outside the range of current
theory, but new theory under development \cite{blanchet15} will 
extend the property to a larger class of L{\'e}vy processes 
including both VG and NIG.

We now bound the weak convergence of the estimator and the multilevel
variance convergence.


\begin{proposition}
\label{prop:barrier} Let $X$ be a scalar L\'{e}vy process underlying an
exponential L\'{e}vy model. For the up-and-out barrier option payoff (\ref%
{def:barrier}), with the numerical approximation (\ref{def:appr_barr}), we
have the following rates of convergence for the multilevel correction
variance and the weak error, assuming that $m$ has a bounded density:

\begin{itemize}
\item If $X$ is a Variance Gamma (VG) process, then
\begin{eqnarray*}
V_\ell &=&\bo(h_{\ell }^{1-\delta }); \\
\left\vert \mathbb{E}\left[ \hP-P\right] \right\vert &=&\bo(h_{\ell
}^{1-\delta })
\end{eqnarray*}%
where $\delta $ is an arbitrary positive number.

\item If $X$ is a NIG process, then
\begin{eqnarray*}
V_\ell &=&\bo(h_{\ell }^{1/2-\delta }); \\
\left\vert \mathbb{E}\left[ \hP-P\right] \right\vert &=&\bo(h_{\ell
}^{1/2-\delta })
\end{eqnarray*}%
where $\delta $ is an arbitrary positive number.

\item If $X$ is a spectrally negative $\alpha $-stable process with $%
\alpha >1$, then%
\begin{eqnarray*}
V_\ell &=&\bo\left( h_{\ell }^{\frac{1}{\alpha }-\delta }\right) ; \\
\left\vert \mathbb{E}\left[ \hP-P\right] \right\vert &=&\bo\left( h_{\ell }^{%
\frac{1}{\alpha }-\delta }\right)
\end{eqnarray*}%
where $\delta $ is an arbitrary positive number.
\end{itemize}
\end{proposition}

\begin{proof}
The variance of the multilevel correction term is bounded by
\begin{equation*}
V_\ell\ \leq \ \mathbb{E}\left[ ( \hP_{\ell }-\hP_{\ell -1}) ^{2}%
\right]\ \leq\ 2\ \mathbb{E}\left[ ( \hP_{\ell }-P) ^{2}\right]
+2\ \mathbb{E}\left[ ( \hP_{\ell -1}-P) ^{2}\right] .
\end{equation*}
For an up-and-out Barrier option, since the payoff is bounded we have
\begin{eqnarray*}
\mathbb{E}\left[ ( \hP_\ell-P) ^{2}\right] &\leq & F^2\ \mathbb{E}%
\left[ \mathbf{1}_{\{\widehat{m}_{n}<\log(B/S_0)\}}-\mathbf{1}_{\left\{ m<\log(B/S_0)\right\} }%
\right] , \\
\left\vert \mathbb{E}[ \hP_\ell-P] \right\vert &\leq & F \
\mathbb{E}\left[ \mathbf{1}_{\{\widehat{m}_{n}<\log(B/S_0)\}}-\mathbf{1}_{\left\{
m<\log(B/S_0)\right\} }\right],
\end{eqnarray*}
where $n=M^\ell$.

The bounds for the VG process come from Proposition \ref{prop:VG-barrier}
together with the results from Theorem \ref{prop:Un}.

The bounds for the NIG come from  taking $p\!=\!1$ in Proposition \ref{prop:avikainen} 
  together with Chen's result Theorem \ref{thm:mean}. 

The bounds for the spectrally
negative $\alpha $-stable process come from Proposition \ref{prop:avikainen}
together with the results from Theorem \ref{prop:Un}. Theorem \ref{prop:Un} gives
\begin{equation*}
\| m \!-\! \widehat{m} \|_p^{p/(p+1)} \ \equiv \ \left( \mathbb{E}\left[\,
|m \!-\! \widehat{m}|^p \right] \right)^{1/(p+1)} \ = \ \bo( h^{\frac{p}{%
(p+1)\alpha} - \frac{\delta}{p+1}} ).
\end{equation*}
We then obtain the desired bound by taking $p$ to be sufficiently large.
\end{proof}

\section{Numerical results}

We have numerical results for three different L\'{e}vy models: Variance
Gamma, Normal Inverse Gaussian and $\alpha $-stable processes, and three
different options: Asian, lookback and barrier. 

The current code is based on Giles' MATLAB code \cite{giles08b}, using 
which we generate standardised numerical results and a set of four figures.
The top two plots correspond to a set of experiments to investigate how
the variance and mean for both $\hP_{\ell}$ and $\hP_{\ell }\!-\!\hP_{\ell-1}$
vary with level $\ell$.
The top left plot shows the values for $\log_2(\mbox{variance})$,
so that the absolute value of the slope of the line for
$\log_2 \VV[\hP_{\ell }\!-\!\hP_{\ell-1}]$ indicates the convergence rate 
$\beta$ of $V_\ell$ in condition i) of Thereom \ref{thm:cc}.  
Similarly, the absolute value of the slope of the line for
$\log_2 |\EE[\hP_{\ell }\!-\!\hP_{\ell-1}]|$ in the top right plot
indicates the weak convergence rate $\alpha$ in the condition i) of 
Thereom \ref{thm:cc}. 

The bottom two plots correspond to a set of MLMC calculations for
different values of the desired accuracy $\eps$.
Each line in the bottom left plot corresponds to one multilevel 
calculation and displays the number of samples $N_\ell$ on each level. 
Note that as $\eps$ is varied, the MLMC algorithm automatically decides 
how many levels are required to reduce the weak error appropriately.  
The optimal number of samples on each level is based on an empirical 
estimation of the multilevel correction variance $V_\ell$, together 
with the use of a Lagrange multiplier to determine how best to minimise 
the overall computational cost for a given target accuracy. 
A complete description of the algorithm is given in \cite{giles15}. 
The bottom right plots show the variation of the computational complexity 
$C$ with the desired accuracy $\eps$. In the best cases, the MLMC complexity 
is $\bO(\eps^{-2})$, and therefore the plot is of $\eps^2\,C$ versus $\eps$ 
so that we can see whether this is achieved, and compare the complexity 
to that of the standard Monte Carlo method.

\begin{figure}[t!]
\begin{center}
\includegraphics[width=0.9\textwidth]{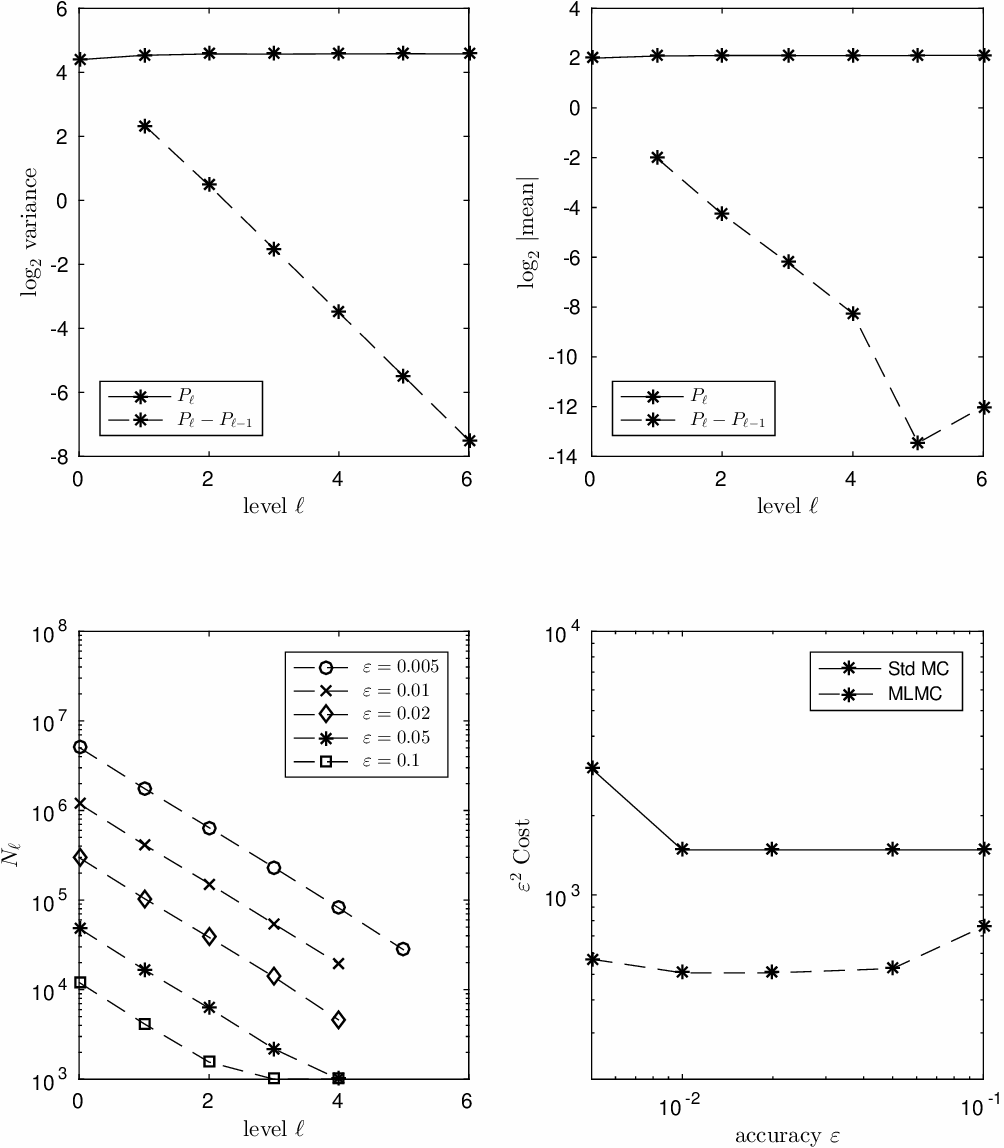}
\end{center}
\par
\vspace{-.1in}
\caption{Asian option in variance gamma model}
\label{vg_a}
\end{figure}

\begin{figure}[t!]
\begin{center}
\includegraphics[width=0.9\textwidth]{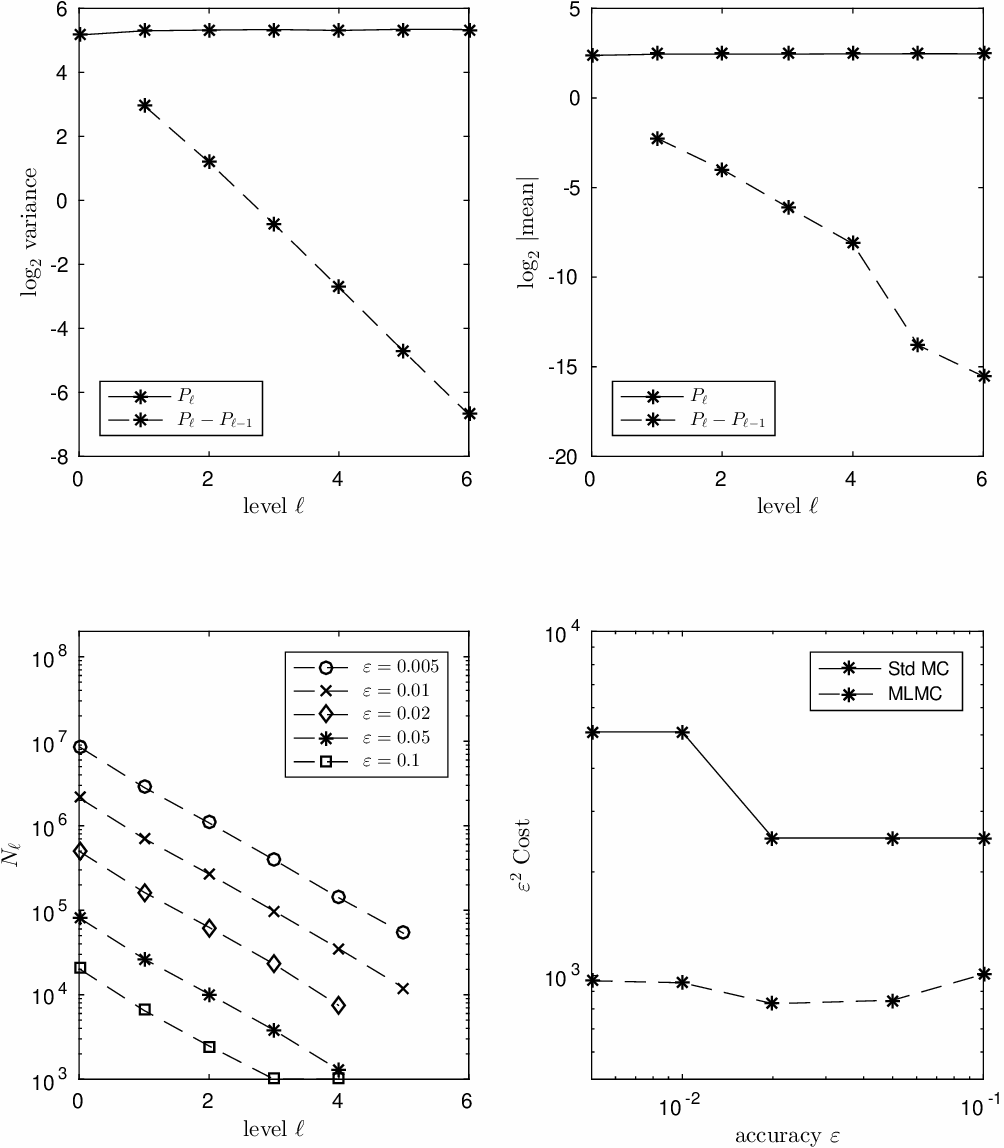}
\end{center}
\par
\vspace{-.1in}
\caption{Asian option in Normal Inverse Gaussian model}
\label{nig_a}
\end{figure}

\begin{figure}[t!]
\begin{center}
\includegraphics[width=0.9\textwidth]{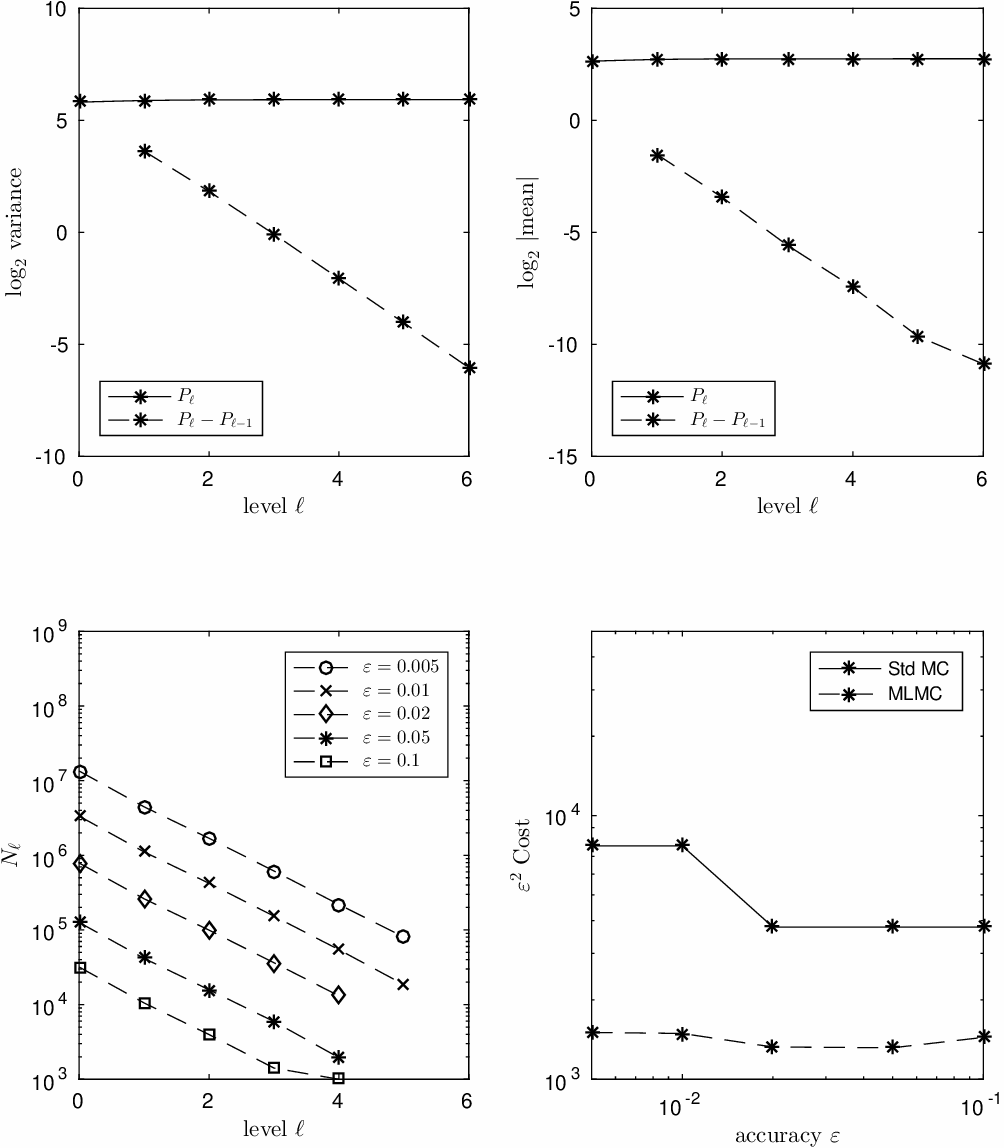}
\end{center}
\par
\vspace{-.1in}
\caption{Asian option in spectrally negative $\protect\alpha $-stable model}
\label{alpha_a}
\end{figure}

\subsection{Asian option}

The Asian option we consider is an arithmetic Asian call option with 
discounted payoff
\begin{equation*}
P=\exp (-rT)\ \max \left( 0,\!\ \overline{S}\!-K\right),
\end{equation*}%
where $T\!=\!1$, $r\!=\!0.05$, $S_{0}\!=\!100$, $K\!=\!100$ and 
\begin{equation*}
\overline{S}=S_{0}\ T^{-1}\int_{0}^{T}\exp \left( X_{t}\right) \D t\ .
\end{equation*}
For a general L\'{e}vy process it is not easy to directly sample the
integral process. We use the trapezoidal approximation%
\begin{equation*}
\overline{\hS}:=S_{0}\ T^{-1} \sum_{j=0}^{n-1}{\textstyle\frac{1}{2}}\,h\,\,(\exp
\left( X_{jh}\right) \!+\!\exp \left( X_{\left( j+1\right) h}\right) ),
\end{equation*}%
where $n=\!T/h$ is the number of timesteps. The payoff approximation
is then
\begin{equation*}
\hP=\exp (-rT)\ \max( 0,\!\ \overline{\hS}\!-K\ ) .
\end{equation*}
In the multilevel estimator, the approximation $\hP_{\ell}$ on level 
$\ell$ is obtained using $n_\ell\!:=\!2^{\ell}$ timesteps. 

Figures \ref{vg_a}, \ref{nig_a}, \ref{alpha_a} are for the 
VG, NIG and $\alpha$-stable models respectively. 
The numerical results in the top right plots indicate approximately
second order weak convergence.  With the standard Monte Carlo method, 
the top left plots show that the variance is approximately independent,
and therefore, the standard Monte Carlo calculation has computational 
cost $\bO(\eps^{-2} n_{\ell}) = \bO(\eps^{-2.5})$.  Multiplying this cost 
by $\eps^{2}$ to create the bottom right complexity plots, the scaled 
cost is $\bO(n_{\ell})$ and therefore goes up in steps as $\eps$ is reduced,
when decreasing $\eps$ requires an increase in the value of the finest level $L$.  
On the other hand, the convergence rate of the variance of the MLMC estimator
is approximately $1.2$ for VG, $2.0$ for NIG and $2$ for the $\alpha$-stable model.
Since in all three cases we have $\beta\!>\!1$, the MLMC theorem gives a
complexity which is $\bO(\eps^{-2})$ which is consistent with the results 
in the bottom right plots which show little variation in $\eps^2\, C$ for 
the MLMC estimator.

For this Asian option, MLMC is 3-8 times more efficient than standard MC.  
The gains are modest because the high rate of weak convergence means that 
only 4 levels of refinement are required in most cases, so there is only 
a $2^4\!=\!16$ difference in cost between each MC path calculation on the 
finest level, and each of the MLMC path calculations on the coarsest level.


\begin{figure}[t!]
\begin{center}
\includegraphics[width=0.9\textwidth]{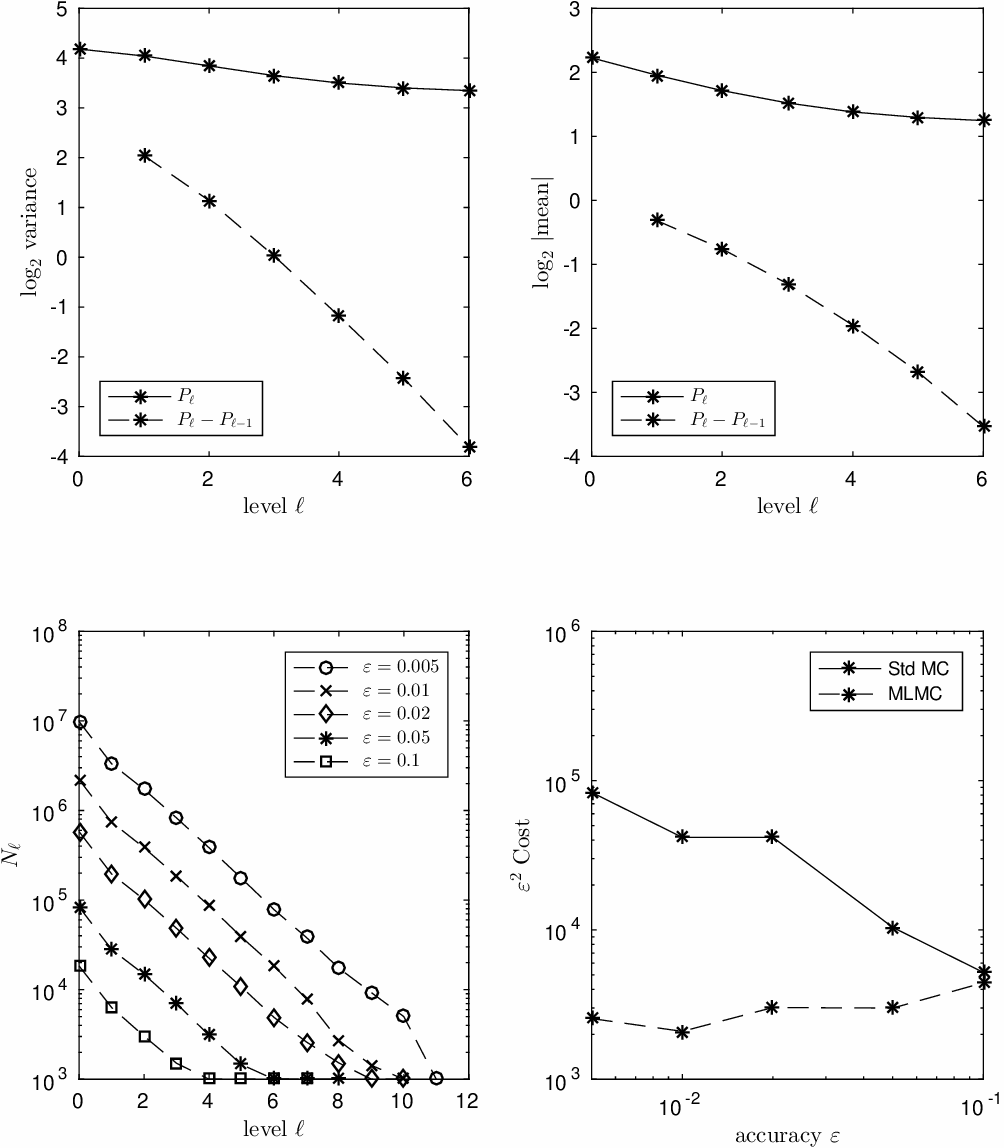} 
\end{center}
\par
\vspace*{-0.25in}
\caption{Lookback option with Variance Gamma model}
\label{vg_l}
\end{figure}

\begin{figure}[t!]
\begin{center}
\includegraphics[width=0.9\textwidth]{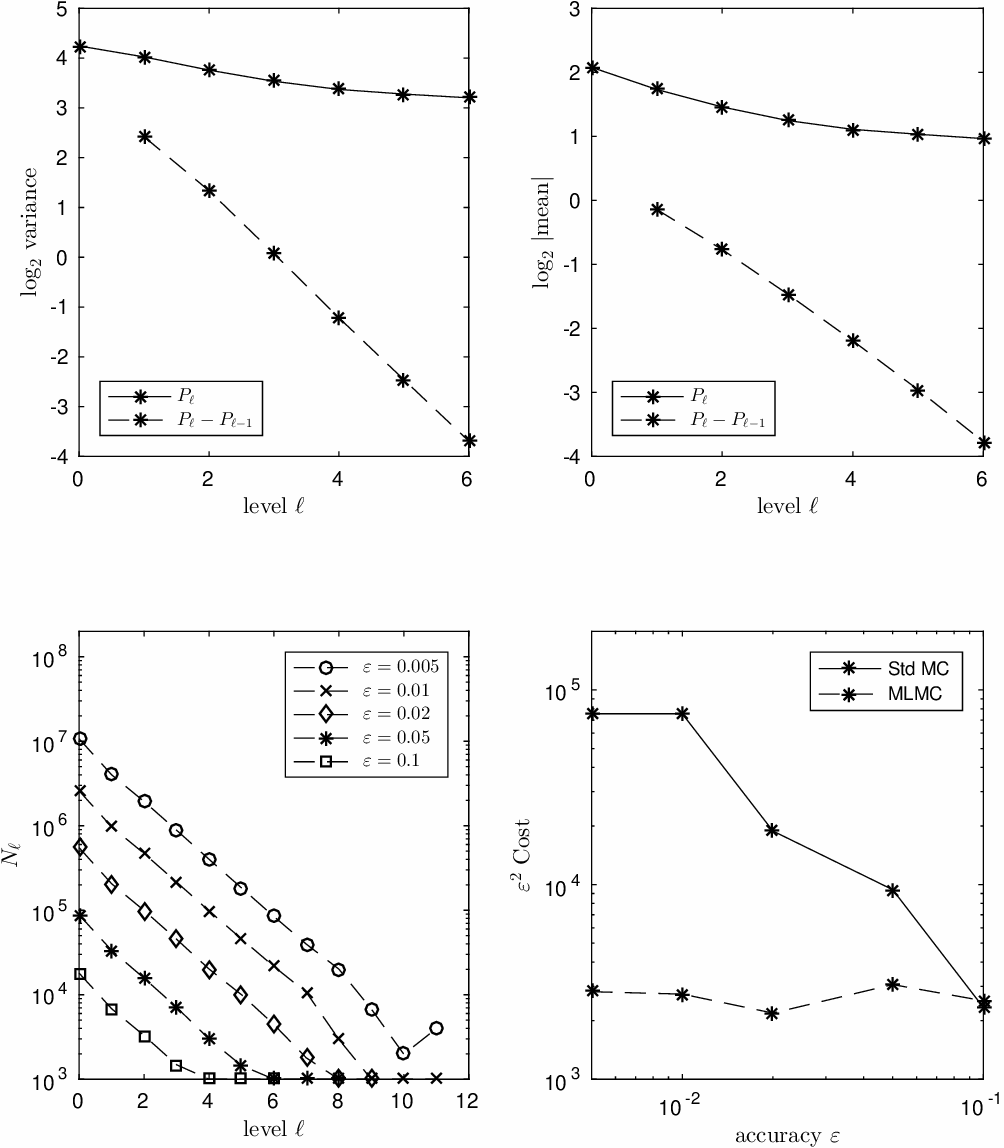}
\end{center}
\par
\vspace*{-0.25in}
\caption{Lookback option with Normal Inverse Gaussian model}
\label{nig_l}
\end{figure}

\begin{figure}[t!]
\begin{center}
\includegraphics[width=0.9\textwidth]{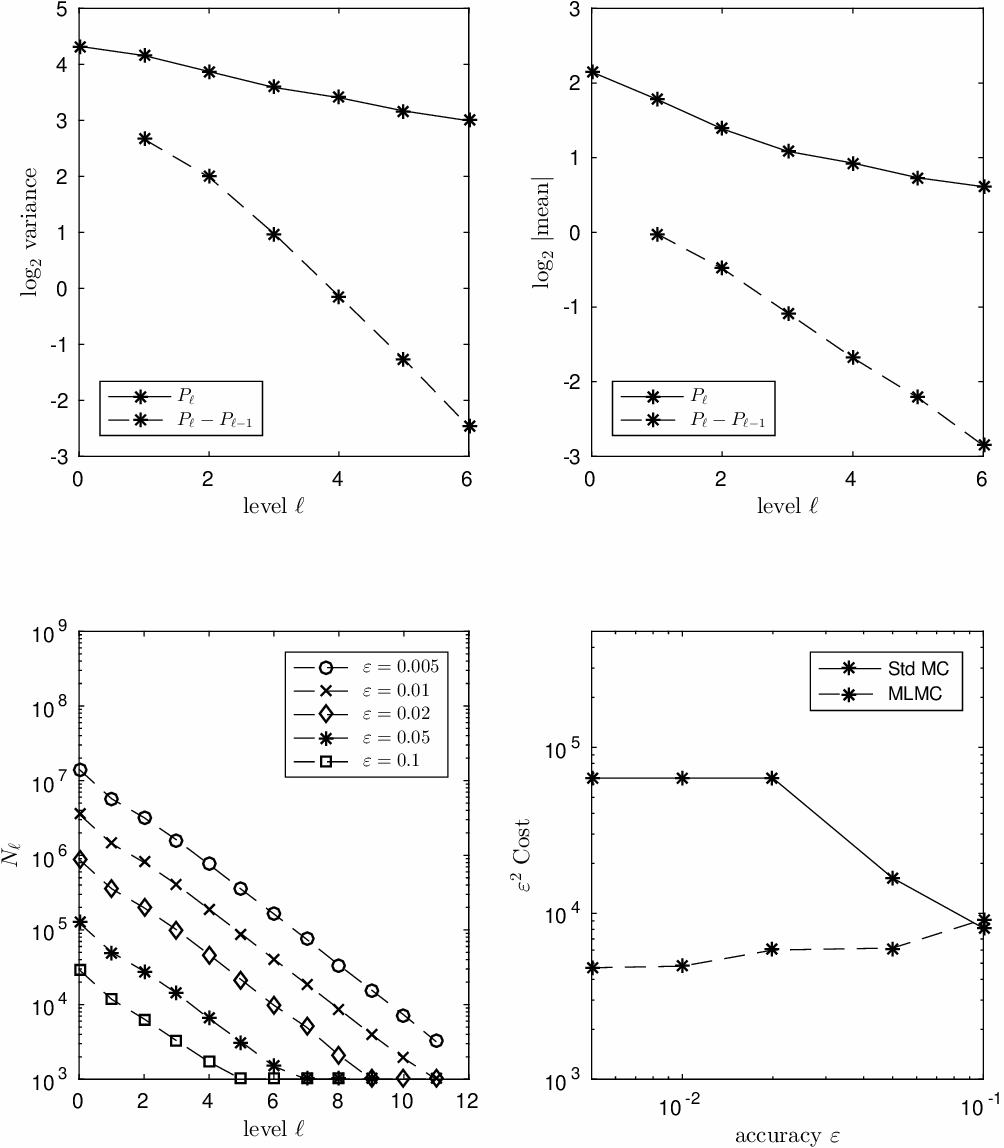}
\end{center}
\par
\vspace*{-0.25in}
\caption{Lookback option with spectrally negative $\protect\alpha $-stable
model}
\label{alpha_l}
\end{figure}

\subsection{\protect\bigskip Lookback option}

The lookback option we consider is a put option on the floating underlying,
\begin{equation*}
P\ =\ \exp (-rT)( K-\sup_{0\leq t\leq T}S_t) ^{+} \! =\ \exp
(-rT)\left( K-S_0\, \exp(m)\, \right) ^{+},
\end{equation*}
where $m=\sup_{0\leq t\leq T}X_t$, 
with $T\!=\!1$, $r\!=\!0.05$, $S_{0}\!=\!100$, $K\!=\!110$. 
We use the discretely monitored maximum as the approximation, so that
\begin{equation*}
\hP_\ell = \exp (-rT)\left( K-S_{0}\, \exp(\widehat{m}_\ell)\, \right) ^{+},
~~~~~ \widehat{m}_\ell=\max_{0\leq j \leq n_\ell} X_{j h_\ell}.
\end{equation*}

Figures \ref{vg_l}, \ref{nig_l}, \ref{alpha_l} show the numerical results for
the VG, NIG and $\alpha$-stable models.  The most obvious difference compared 
to the Asian option is a greatly reduced order of weak convergence, approximately
$1$, $0.8$ and $0.6$ in the respective cases.  This reduced weak convergence leads
to a big increase in the finest approximation level, which in turn greatly 
increases the standard MC cost but doesn't significantly change the MLMC cost.
Hence, the computational savings are much greater than for the Asian option,
with savings of up to a factor of 30.

The small erratic fluctuation in $N_{\ell }$ on levels greater than $5$ is due to 
poor estimates of the variance due to a limited number of samples.  This also 
appears later for the barrier option.

\begin{figure}[t!]
\begin{center}
\includegraphics[width=0.9\textwidth]{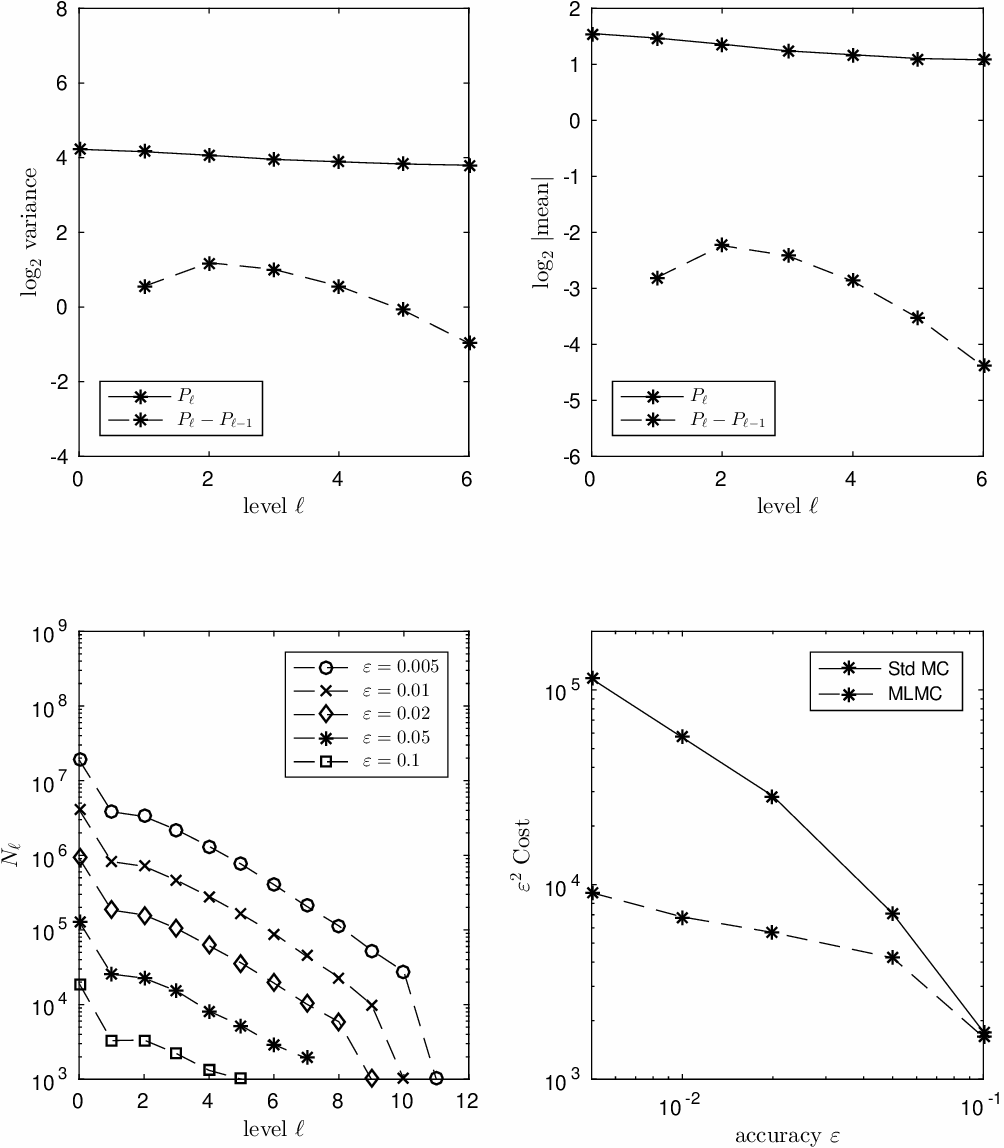}
\end{center}
\par
\vspace*{-0.25in}
\caption{Barrier option in variance gamma model}
\label{vg_b}
\end{figure}

\begin{figure}[t!]
\begin{center}
\includegraphics[width=0.9\textwidth]{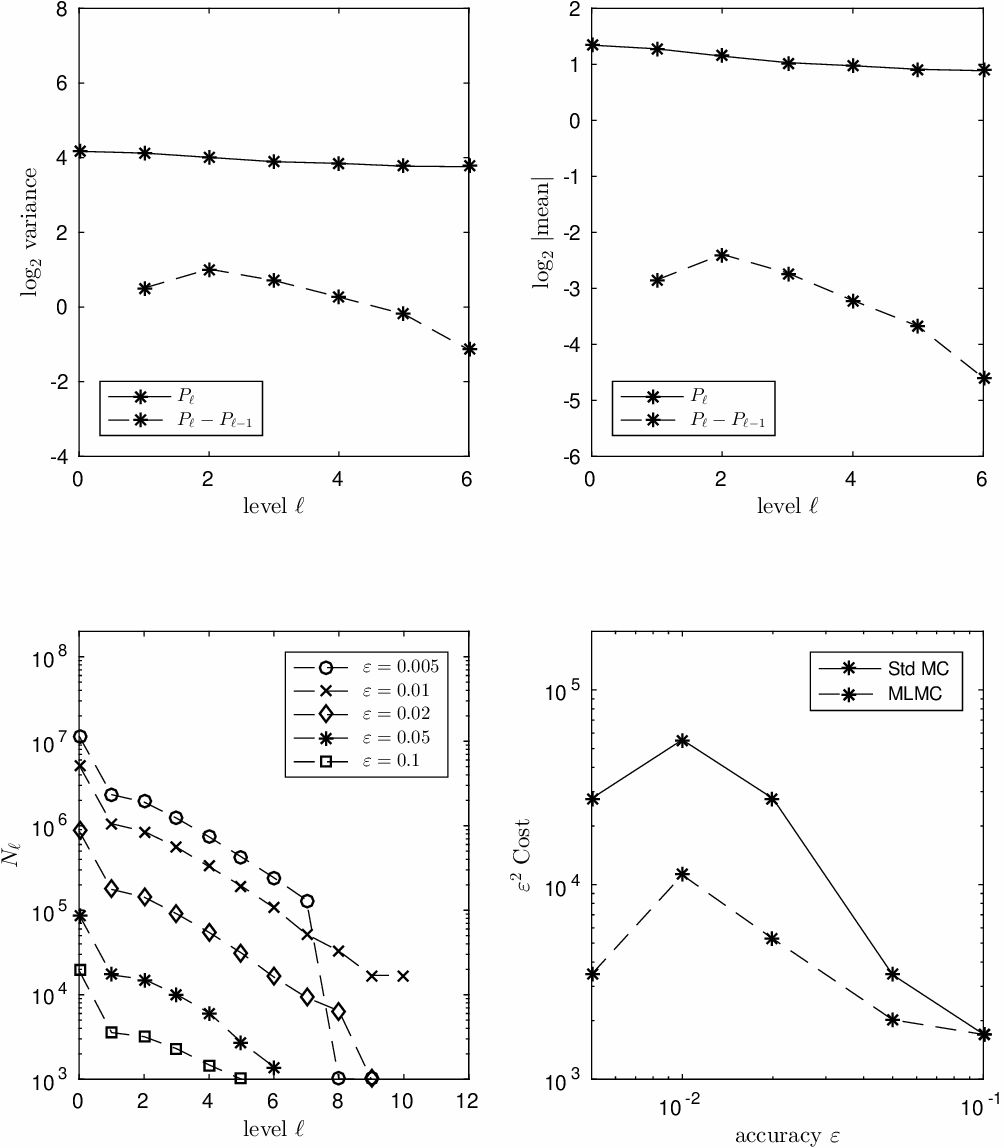}
\end{center}
\par
\vspace*{-0.25in}
\caption{Barrier option in Normal Inverse Gaussian model}
\label{nig_b}
\end{figure}

\subsection{Barrier option}

The barrier option is an up-and-out call with payoff
\begin{equation*}
P\ =\ \exp (-rT)\,(S_{T}\!-\!K)^{+}\,\ii_{\left\{ \sup_{0\leq t\leq T}S(t)<B\right\} }
 \ =\ \exp (-rT)\,(S_{T}\!-\!K)^{+}\,\ii_{\left\{ m < \log(B/S_0)\right\} },
\end{equation*}%
with $T\!=\!1$, $r\!=\!0.05$, $S_0\!=\!100$, $K\!=\!100$, $B\!=\!115$.
The discretely monitored approximation is
\begin{equation*}
\hP_\ell = \exp (-rT)\,(S_{T}\!-\!K)^{+}\,\ii_{\left\{ \hat{m_\ell}%
<\log(B/S_0)\right\} }, 
~~~~~ \widehat{m}_\ell=\max_{0\leq j \leq n_\ell} X_{j h_\ell}
\end{equation*}


With the barrier option, the most noticeable change from the previous options 
is a reduction in the rate of convergence $\beta$ of the MLMC variance, with 
$\beta \approx 0.75, 0.5, 0.6$ in the three cases.  For $\beta\!<\!1$, the MLMC 
theorem proves a complexity which is $\bO(\eps^{-2 - (1\!-\!\beta)/\alpha})$, with 
$\alpha$ here being the rate of weak convergence.  The fact that the MLMC 
complexity is not $\bO(\eps^{-2})$ is clearly visible from the bottom right
complexity plots, but there are still significant savings compared to the standard 
MC computations.


\begin{figure}[t]
\begin{center}
\includegraphics[width=0.9\textwidth]{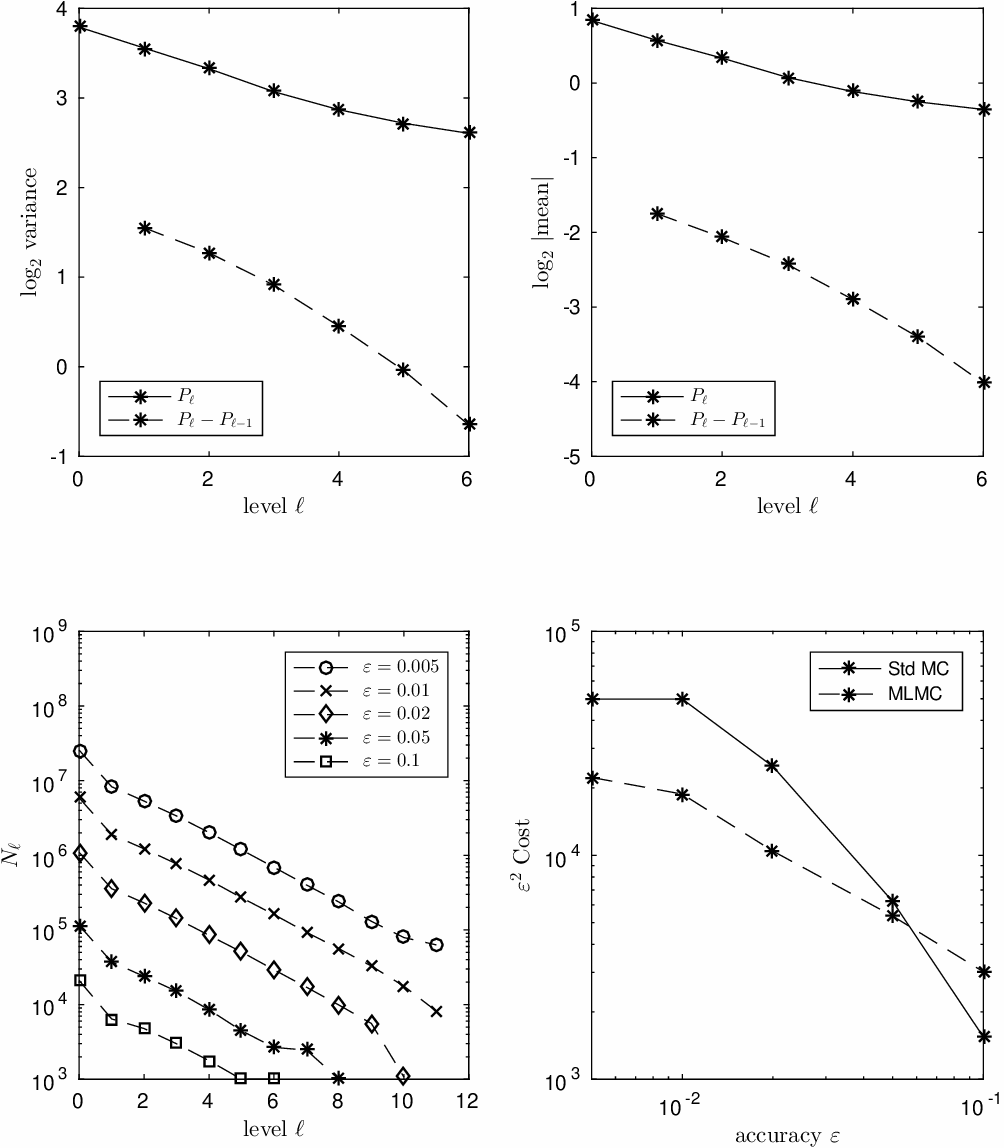}
\end{center}
\par
\vspace*{-0.25in}
\caption{Barrier option in spectrally negative $\protect\alpha $-stable
model }
\label{alpha_b}
\end{figure}

\newpage

\subsection{Summary and discussion}


Table \ref{table:levy_rate} summarizes the convergence 
rates for the weak error $\EE[\hP_{\ell}\!-\!P]$ and 
the MLMC variance $V_{\ell}=\VV[\hP_{\ell}\!-\!\hP_{\ell-1}]$ 
given by Propositions \ref{prop:arith_asian}, \ref{prop:lookback}, 
\ref{prop:barrier}, and the empirical convergence rates
observed in the numerical experiments.

In general, the agreement between the analysis and the numerical rates 
of convergence is quite good, suggesting that in most cases the analysis 
may be sharp.
The most obvious gap between the two is with the weak order of 
convergence for the Asian option with all three models; the analysis
proves an $\bO(h)$ bound, whereas the numerical results suggest it 
is actually $\bO(h^2)$. The numerical results are perhaps not surprising 
as $\bO(h^2)$ is the order of convergence of trapezoidal integration 
of a smooth function, and therefore it is the order one would expect 
if the payoff was simply a multiple of $\overline{S}$.

\begin{table}[tbp]
\caption{Convergence rates of weak error and variance $V_{\ell}$ for 
VG, NIG and $\protect\alpha $-stable processes; $\protect\delta $ 
can be any small positive constant. The numerical values are estimates
based on the numerical experiments.}
\label{table:levy_rate}
\par
\begin{center}
\begin{tabular}{|l|l|l|l|l|}
\hline
& \multicolumn{4}{c|}{VG} \\
& \multicolumn{2}{|c}{numerical} & \multicolumn{2}{|c|}{analysis} \\
option & weak & var & weak & var \\ \hline
Asian & $\bO\left( h^2\right) $ & $\bO\left( h^{2}\right) $ & $\bO\left(
h\right) $ & $\bO\left( h^{2}\right) $ \\
lookback & $\bO\left( h\right) $ & $\bO\left( h^{1.2}\right) $ & $\bO\left(
h\left\vert \log h\right\vert \right) $ & $\bO\left( h\right) $ \\
barrier & $\bO\left( h^{0.8}\right) $ & $\bO\left( h^{0.9}\right) $ & $\bo\left(
h^{1-\delta }\right) $ & $\bo\left( h^{1-\delta }\right) $ \\ \hline
\end{tabular}%
\par
\vspace{0.1in}
\begin{tabular}{|l|l|l|l|l|}
\hline
& \multicolumn{4}{c|}{NIG} \\
& \multicolumn{2}{|c|}{numerical} & \multicolumn{2}{|c|}{analysis} \\
option & weak & var & weak & var \\ \hline
Asian & $\bO\left( h^2\right) $ & $\bO\left( h^{2}\right) $ & $\bO\left(
h\right) $ & $\bO\left( h^{2}\right) $ \\
lookback & $\bO\left( h^{0.8}\right) $ & $\bO\left( h^{1.2}\right) $ & $\bo%
\left( h^{1-\delta }\right) $ & $\bO\left( h\left\vert \log h\right\vert
\right) $ \\
barrier & $\bO\left( h^{0.4}\right) $ & $\bO\left( h^{0.5}\right) $ & $\bo%
\left( h^{0.5-\delta }\right) $ & $\bo\left( h^{0.5-\delta }\right) $ \\
\hline
\end{tabular}%
\par
\vspace{0.1in}
\begin{tabular}{|l|l|l|l|l|}
\hline
& \multicolumn{4}{c|}{spectrally negative $\alpha $-stable with $\alpha >1$}
\\
& \multicolumn{2}{|c}{numerical for $\alpha =1.5597$} & \multicolumn{2}{|c|}{
analysis} \\
option & weak & var & weak & var \\ \hline
Asian & $\bO\left( h^2\right) $ & $\bO\left( h^{2}\right) $ & $\bO\left(
h\right) $ & $\bO\left( h^{2}\right) $ \\
lookback & $\bO\left( h^{0.6}\right) $ & $\bO\left( h^{1.6}\right) $ & $\bo%
\left( h^{1/\alpha -\delta }\right) $ & $\bo\left( h^{2/\alpha -\delta
}\right) $ \\
barrier & $\bO\left( h^{0.5}\right) $ & $\bO\left( h^{0.6}\right) $ & $\bo%
\left( h^{1/\alpha -\delta }\right) $ & $\bo\left( h^{1/\alpha -\delta
}\right) $ \\ \hline
\end{tabular}%
\end{center}
\end{table}

\section{Proofs}
\label{sec:proofs}

\subsection{Proof of Proposition \protect\ref{prop:arith_asian}}
\label{sec:proofs-5.1}

\begin{proof}
We decompose the difference between the true value and approximation 
into parts which we can bound separately:
\begin{eqnarray*}
\left\vert \overline{S}-\overline{\hS}\right\vert &=& S_0 \ T^{-1}\left\vert
\int_{0}^{T}\exp \left(X_{t}\right) \D t\ - \ 
\sum_{j=0}^{n-1}{\textstyle\frac{1}{2}}\,h\,\,(\exp \left( X_{jh}\right)
\!+\!\exp \left( X_{\left( j+1\right) h}\right) )
\right\vert \\
&=& S_0 \ T^{-1}\Bigg\vert\sum_{j=0}^{n-1}{\textstyle}\exp \left( X_{jh}\right)
\int_{jh}^{\left( j+1\right) h}\left( \exp \left( X_{t}-X_{jh}\right)
-1\right) \D t-{\fracs{1}{2}h}\exp \left( X_{T}\right) +{\fracs{1}{2}}h%
\Bigg\vert.
\end{eqnarray*}%
If we define
\begin{eqnarray*}
{b}_{j} &=&\exp \left( X_{jh}\right), \\
I_{j} &=&\int_{jh}^{(j+1) h}\left( \exp \left(
X_{t}\!-\!X_{jh}\right) -1\right)\ \D t, \\
R_{A} &=&-{\fracs{1}{2}h}\exp \left( X_{T}\right) +{\fracs{1}{2}}h,
\end{eqnarray*}
then
\[
\mathbb{E}\left[ \left( \overline{\hS}-\overline{S}\right) ^{2}\right]
\ =\ 
T^{-2}S_{0}^{2}\ \mathbb{E}\left[ \left\vert 
\sum_{j=0}^{n-1}{b}_{j}I_{j}+R_{A}\right\vert ^{2}\right]
\ \leq\ 
2T^{-2}S_{0}^{2}\left( \mathbb{E}\left[ \left\vert 
\sum_{j=0}^{n-1}{b}_{j}I_{j}\right\vert ^{2}\right] +\mathbb{E}\left[ R_{A}^{2}\right] \right) .
\]
We have $\mathbb{E}\left[ R_{A}^{2}\right] =\bO\left( h^{2}\right)$, and
due to the independence of ${b}_{j}$ and $I_{j}$ we obtain
\begin{eqnarray}
\mathbb{E}\left[ \left\vert \sum_{j=0}^{n-1}{b}_{j}I_{j}\right\vert ^{2}%
\right] &=&\mathbb{E}\left[ \sum_{j=0}^{n-1}{b}_{j}^{2}I_{j}^{2}+2%
\sum_{m=1}^{n-1}\sum_{j=0}^{m-1}{b}_{m}I_{m}{b}_{j}I_{j}\right]  \notag \\
&=&\sum_{j=0}^{n-1}\mathbb{E}\left[ {b}_{j}^{2}\right] \mathbb{E}\left[
I_{j}^{2}\right] +2 \sum_{m=1}^{n-1}\sum_{j=0}^{m-1}
\mathbb{E}\left[{b}_{m}I_{m}{b}_{j}I_{j}\right].  \label{arith_asian}
\end{eqnarray}%
Defining
$A=2m+\int \left(
e^{2z}-1-2z1_{\left\vert z\right\vert <1}\right) \nu \left( \D z\right)$,
we have $\mathbb{E}\left[ {b}_{j}^{2}\right]=e^{Ajh}$.
Furthermore, by the Cauchy-Schwarz inequality,%
\begin{eqnarray*}
\mathbb{E}\left[ I_{j}^{2}\right]
&\leq & h \ \mathbb{E}\left[  \int_{jh}^{\left( j+1\right) h}
\left( \exp \left(
X_{t}-X_{jh}\right) -1\right) ^{2}\D t\right] \\
&=& h \int_0^h \mathbb{E}\left[\left( \exp \left( X_{t}\right) -1\right)
^{2} \right]\, \D t\\
&=& h \left( \frac{1}{A}\left( e^{Ah}-1-Ah\right) - 
2\frac{1}{r}\left(e^{rh}-1-rh\right) \right)
\end{eqnarray*}%

Note that $1+x < e^x < 1+x+x^2$ for $0\!<\!x\!<\!1$, and therefore
for $h\!<\!1/A$ we have $\mathbb{E}\left[ I_{j}^{2}\right]<Ah^3$ and hence
\[
\sum_{j=0}^{n-1}\mathbb{E}\left[ {b}_{j}^{2}\right] \mathbb{E}\left[I_{j}^{2}\right] 
\ <\  A\,  h^{3}\, \sum_{j=0}^{n-1}e^{Ajh} 
\ =\ A \, \frac{e^{AT}-1}{e^{Ah}-1}\ h^3
\ <\  (e^{AT}\!-\!1)\ h^2.
\]

Now we calculate the second term in (\ref{arith_asian}). Note that
for $m\!>\!j,$ $I_{m}$ is independent of ${b}_{m}{b}_{j}I_{j}$, and 
${b}_{m}/b_{j+1}$ is independent of $b_{j+1}{b}_{j}I_{j},$ so%
\begin{equation*}
\sum_{m=1}^{n-1}\sum_{j=0}^{m-1}
\mathbb{E}\left[ {b}_{m}I_{m}{b}_{j}I_{j}%
\right] =\sum_{m=1}^{n-1}\mathbb{E}\left[ I_{m}\right] \sum_{j=0}^{m-1}%
\mathbb{E}\left[ {b}_{m}/b_{j+1}\right] \mathbb{E}\left[ b_{j+1}{b}_{j}I_{j}%
\right] .
\end{equation*}
Firstly, for $h<1/r$,
\[
\mathbb{E}\left[ I_{m}\right]
\ =\ \int_{0}^{h}(e^{rt} - 1)\ \D t
\ =\ r^{-1}\left( e^{rh}-1 - rh\right)
\ < \ r\, h^2.
\]
Moreover,  we have
$\mathbb{E}\left[ {b}_{m}/b_{j+1}\right] =e^{r\left( m-j-1\right) h}$
and
\begin{eqnarray*}
\mathbb{E}\left[ b_{j+1}{b}_{j}I_{j}\right]  
&\!=\!&\mathbb{E}\left[ \exp \left(
2X_{jh}\right) \exp \left( X_{\left( j+1\right) h}\!-\!X_{jh}\right)
\int_{jh}^{\left( j+1\right) h} \Big(\exp \left( X_{t}\!-\!X_{jh}\right)
-1\Big)\  \D t\right]  \\
&\!=\!&\mathbb{E}\left[ \exp \left( 2X_{jh}\right) \right] 
\ \mathbb{E}\left[ \exp \left( X_{h}\right) \int_{0}^{h}\Big(
\exp \left( X_{t}\right) -1\Big)\ \D t\right]  \\
&\!=\!&e^{Ajh}\int_{0}^{h}\Big(
\mathbb{E}\left[ \exp \left( X_{h}\!-\!X_{t}\right) \right]
\mathbb{E}\left[ \exp \left( 2X_{t}\right) \right] 
-\mathbb{E}\left[\exp\left( X_{h}\right)\right]\Big) \ \D t \\
&\!=\!&e^{Ajh}\int_{0}^{h} \left(e^{r(h-t)} e^{At} - e^{rh} \right) \ \D t \\
&\!=\!&e^{Ajh}e^{rh}\ 
\frac{e^{(A-r)h}-1-(A\!-\!r)h}{A\!-\!r}.
\end{eqnarray*}%
Thus, for $h<1/(A\!-\!r)$,
\begin{eqnarray*}
\sum_{m=1}^{n-1}\sum_{j=0}^{m-1}\mathbb{E}\left[ {b}_{m}/b_{j+1}\right]
\mathbb{E}\left[ b_{j+1}{b}_{j}I_{j}\right]
&=& \frac{e^{(A-r)h}-1-(A\!-\!r)}{A\!-\!r}\ 
\sum_{m=1}^{n-1}\sum_{j=0}^{m-1}e^{r(m-j)h}e^{Ajh} \\
&=& \frac{e^{(A-r)h}-1-(A\!-\!r)h}{(A\!-\!r)\, (e^{(A-r) h}-1)}\ 
\sum_{m=1}^{n-1} (e^{Amh}-e^{rmh})  \\
&<& h \, \frac{e^{AT}-1}{e^{Ah}-1} \\[0.05in]
&<& A^{-1} (e^{AT}-1).
\end{eqnarray*}%
Hence,
\[
\mathbb{E}\left[ \sum_{m=1}^{n-1}\sum_{j=0}^{m-1}b_m I_m b_j I_j \right] 
= \sum_{m=1}^{n-1}\mathbb{E} \left[ I_m\right] \sum_{j=0}^{m-1}
\mathbb{E}\left[ b_m/b_{j+1}\right] 
\mathbb{E}\left[ b_{j+1}b_j I_j \right]
=\bO(h^{2}),
\]
and we can therefore conclude that 
$\displaystyle
\mathbb{E}\left[ \left( \overline{\hS}-\overline{S}\right) ^{2}\right] =\bO%
(h^{2}).
$
\end{proof}

\subsection{ L\'{e}vy process decomposition}

\label{sec:decomposition}

The proofs rely on a decomposition of the L\'{e}vy process into a
combination of a finite-activity pure jump part, a drift part, and a
residual part consisting of very small jumps.

Let $X$ be an $(m,0,\nu )$-L\'{e}vy process:%
\begin{equation}
X_{t}=mt+\int_{0}^{t}\int_{\{ |z| \geq 1\}}z\ J(\D z,\D %
s)+\int_{0}^{t}\int_{\{ |z| <1\}}z\left( J(\D z,\D s)-\nu (\D %
z)\D s\right) .  \label{pure-jump-decomposition}
\end{equation}%
The finite activity jump part is defined by
\begin{equation*}
X_{t}^{\varepsilon }=\int_{0}^{t}\int_{\{\varepsilon <\left\vert z\right\vert
\}}z\ J(\D z,\D s)=\sum_{i=1}^{N_{t}}Y_{i}
\end{equation*}%
to be the compound Poisson process truncating the jumps of $X$ smaller
than $\varepsilon $ which is assumed to satisfy $0\!<\!\varepsilon \!<\!1$.
The intensity of $N_{t}$ and the c.d.f. of $Y_{i}$ are
\begin{eqnarray}
\lambda _{\varepsilon } &=&\int_{\{\varepsilon <\left\vert z\right\vert \}}\nu (%
\D z).  \label{def:lam} \\
\mathbb{P}\left[ Y_{i}<y\right] &=&\lambda _{\varepsilon
}^{-1}\int_{\{z<y\}}1_{\left\{ \varepsilon <\left\vert z\right\vert \right\}
}\nu (\D z);  \notag
\end{eqnarray}%
The drift rate for the drift term is defined to be
\begin{equation}
\mu _{\varepsilon }=m-\int_{\{\varepsilon <\left\vert z\right\vert <1\}}z\ \nu (%
\D z),  \label{def:mu}
\end{equation}%
so that the residual term is then a martingale:%
\begin{equation}
R_{t}^{\varepsilon }:=\int_{0}^{t}\int_{\{ |z| \leq
\varepsilon \}}z\left( J(\D z,\D s)-\nu (\D z)\D s\right) .  \label{def:R}
\end{equation}%
We define
\begin{equation}
\sigma _{\varepsilon }^{2}=\int_{\{ |z| \leq \varepsilon
\}}z^{2}\nu (\D z),  \label{sig_ep}
\end{equation}%
so that $\VV\left[ R_{t}^{\varepsilon }\right] =\sigma _{\varepsilon }^{2}t$.

These three quantities, $\mu _{\varepsilon }$, $\lambda _{\varepsilon }$ and
$\sigma _{\varepsilon }$ will all play a major role in the subsequent
numerical analysis.

We bound $D_{n}$ by the difference between continuous maxima and 2-point
maxima over all timesteps:
\begin{equation}
D_{n}\ =\ \sup_{0\leq t\leq 1}X_{t}-\max_{i=0,1,\ldots ,n}X_{\frac{i}{n}} \
\leq \ \max_{i=0,\ldots ,n-1} D_{n}^{(i)}  \label{estimate:2-point}
\end{equation}
where the random variables
\begin{equation*}
D_{n}^{(i)} = \sup_{[\frac{i}{n},\frac{i+1}{n}]}X_{t} - \max \left( X_{\frac{%
i+1}{n}},X_{\frac{i}{n}}\right)
\end{equation*}
are independent and identically distributed. If we now define
\begin{equation*}
\Delta^{(i)} X_t = X_{\frac{i}{n}+t} - X_{\frac{i}{n}}, ~~~~ \Delta^{(i)}
X^\varepsilon_t = X^\varepsilon_{\frac{i}{n}+t} - X^\varepsilon_{\frac{i}{n}%
}, ~~~~ \Delta^{(i)} t = t - \frac{i}{n}, ~~~~ \Delta^{(i)} R^\varepsilon_t
= R^\varepsilon_{\frac{i}{n}+t} - R^\varepsilon_{\frac{i}{n}}, ~~~~
\end{equation*}
then
\begin{eqnarray}
D_{n}^{(i)} &=&\sup_{[0,\frac{1}{n}]}\Delta^{(i)} X_{t}-\left( \Delta^{(i)}
X_{\frac{1}{n}}\right) ^{+}  \notag \\
&=&\sup_{[0,\frac{1}{n}]}\left( \Delta^{(i)} X_{t}^{\varepsilon
}+\Delta^{(i)} R_{t}^{\varepsilon }+\mu _{\varepsilon }\Delta^{(i)}t\right)
-\left( \Delta^{(i)} X_{\frac{1}{n}}^{\varepsilon } +\Delta^{(i)} R_{\frac{1%
}{n}}^{\varepsilon }+\mu _{\varepsilon }\frac{1}{n}\right) ^{+}  \notag \\
&\leq &\sup_{[0,\frac{1}{n}]}\left( \Delta^{(i)} X_{t}^{\varepsilon
}+\Delta^{(i)} R_{t}^{\varepsilon }\right) -\left( \Delta^{(i)} X_{\frac{1}{n%
}}^{\varepsilon }+\Delta^{(i)} R_{\frac{1}{n}}^{\varepsilon }\right) ^{+}+%
\frac{\left\vert \mu _{\varepsilon }\right\vert }{n}  \notag \\
&\leq &\sup_{[0,\frac{1}{n}]}\Delta^{(i)} X_{t}^{\varepsilon }-\left(
\Delta^{(i)} X_{\frac{1}{n}}^{\varepsilon }\right) ^{+}+\frac{\left\vert \mu
_{\varepsilon }\right\vert }{n}+\sup_{[0,\frac{1}{n}]}\Delta^{(i)}
R_{t}^{\varepsilon }+\left( -\Delta^{(i)}R_{\frac{1}{n}}^{\varepsilon
}\right) ^{+}  \notag \\
&\leq &\sup_{[0,\frac{1}{n}]}\Delta^{(i)} X_{t}^{\varepsilon }-\left(
\Delta^{(i)} X_{\frac{1}{n}}^{\varepsilon }\right) ^{+}+\frac{\left\vert \mu
_{\varepsilon }\right\vert }{n}+2\sup_{[0,\frac{1}{n}]}\left\vert
\Delta^{(i)} R_{t}^{\varepsilon }\right\vert  \label{estimate:Vp_raw}
\end{eqnarray}%
where we use $\left( a\!+\!b\right) ^{+}\leq a^{+}+b^{+}$ with $%
a=\Delta^{(i)}X_{\frac{1}{n}}^{\varepsilon } +\Delta^{(i)}R_{\frac{1}{n}%
}^{\varepsilon } +\mu _{\varepsilon }\frac{1}{n}$, $b=-\mu _{\varepsilon }%
\frac{1}{n}$ in the first inequality, and $a=\Delta^{(i)}X_{\frac{1}{n}%
}^{\varepsilon } +\Delta^{(i)}R_{\frac{1}{n}}^{\varepsilon }$ , $%
b=-\Delta^{(i)}R_{\frac{1}{n}}^{\varepsilon }$ in the second inequality.

Let $Z^{(i)}_{n}:=\sup_{[0,\frac{1}{n}]} \Delta^{(i)} X_{t}^{\varepsilon } -
\left(\Delta X_{\frac{1}{n}}^{\varepsilon } \right)^+$ and $%
S^{(i)}_{n}:=\sup_{[0,\frac{1}{n}]}\left\vert \Delta^{(i)}
R_{t}^{\varepsilon }\right\vert$. Then, for $p\geq 1$, Jensen's inequality
gives us
\begin{eqnarray}  \label{decomp:Vp}
\lefteqn{\mathbb{E}\left[ D_{n}^{p}\right] }  \notag \\
&\!\!\leq\!\! &\mathbb{E}\left[ \max_{0\leq i<n} ( Z^{(i)}_{n} +\frac{%
\left\vert\mu_\varepsilon\right\vert}{n} +2 \, S^{(i)}_{n} ) ^{p}%
\right]  \notag \\
&\!\!\leq\!\! &3^{p-1} \mathbb{E}\left[ \max_{0\leq i<n} ( Z^{(i)}_{n}
)^{p} + \left( \left\vert \mu_{\varepsilon }\right\vert /n%
\right)^{p} + 2^p\! \max_{0\leq i<n} \left( S_{n}^{(i)} \right)^{p} \right]
\notag \\
&\!\!\leq\!\! & 3^{p-1} n\ \mathbb{E}\Big[( \sup_{[0,\frac{1}{n}%
]}X_{t}^{\varepsilon } -\left( X_{\frac{1}{n}}^{\varepsilon }\right)
^{+})^{p}\Big] \!+ 3^{p-1}\!\!\left( \left\vert \mu_{\varepsilon }\right\vert /n%
\right)^{p} \!\!+3^{p-1}2^p\, \mathbb{E}\Big[
\max_{0\leq i<n} \left( S_n^{(i)} \right)^p \Big]  \notag \\
\end{eqnarray}
where in the final step we have used the fact that all of the $%
\Delta^{(i)}X^\varepsilon_t$ have the same distribution as $X^\varepsilon_t$.

The task now is to bound the first and third terms in the final line of (\ref%
{decomp:Vp}).

\subsection{Bounding moments of {$\sup_{[0,\frac{1}{n}]}X_t^\protect%
\varepsilon-(X_\frac{1}{n}^\protect\varepsilon )^+$}}

\begin{theorem}
\label{thm:main}Let $X$ be a scalar L\'{e}vy process with a triple $%
(m,0,\nu )$, and let $X_t^\varepsilon$, $\mu_\varepsilon$, $%
\lambda_\varepsilon$ and $\sigma_\varepsilon$ be as defined in section \ref%
{sec:decomposition}.

Then provided $\lambda_\varepsilon \!\leq\! n$, for any $p\!>\!1$ there
exists a constant $K_p$ such that
\begin{equation}
\mathbb{E}\Big[ \big( \sup_{[0,\frac{1}{n}]}X_{t}^{\varepsilon } -\left(
X_{\frac{1}{n}}^{\varepsilon }\right) ^{+}\big) ^{p}\Big] \leq
K_{p}\big( \varepsilon ^{p}+\frac{L_{\varepsilon }\left( p\right) }{\lambda
_{\varepsilon }^{2}}\big) \frac{\lambda _{\varepsilon }^{2}}{n^{2}},
\label{estimate_Z}
\end{equation}%
where $L_{\varepsilon }\left( p\right) = p\int_{x>\varepsilon}x^{p-1}\lambda
_{x}^{2}\, \D x$ is a function depending on the L\'{e}vy measure $\nu(x)$.
\end{theorem}

\begin{proof}
Let
\begin{equation*}
Z=\sup_{[0,\frac{1}{n}]}X_t^\varepsilon-\left( X_{\frac{1}{n}%
}^\varepsilon\right)^+.
\end{equation*}

We will determine an upper bound on $\mathbb{E}\left[ Z^{p}\right]$ by
analysing the jump behavior of the finite-activity process $X_t^\varepsilon$
in a single interval $[0,\frac{1}{n}]$.

\begin{figure}[t!]
\begin{center}
{\setlength{\unitlength}{1cm}
\begin{picture}(10.5,4)(-0.5,0)
\put(0,0){\vector(1,0){4.5}}
\put(0,0){\vector(0,1){3}}
\put(6,0){\vector(1,0){4.5}}
\put(6,0){\vector(0,1){3}}
\put(3,-0.2){\makebox(0,0){$t$}}
\put(-0.25,2){\makebox(0,0){$X_t^\varepsilon$}}
\put(8,-0.2){\makebox(0,0){$t$}}
\put(5.75,2){\makebox(0,0){$X_t^\varepsilon$}}

\thicklines
\put(0,0){\line(1,0){2}}
\put(2,2){\line(1,0){2}}
\put(0,0){\circle*{0.15}}
\put(4,2){\circle*{0.15}}
\put(6,0){\line(1,0){1.5}}
\put(7.5,3){\line(1,0){1.5}}
\put(9,2){\line(1,0){1}}
\put(6,0){\circle*{0.15}}
\put(10,2){\circle*{0.15}}

\end{picture}} 
\end{center}
\caption{Behavior of $X_t^\protect\varepsilon$ in the case of one or two
jumps in the interval $[0,\frac{1}{n}]$.}
\label{fig:jump}
\end{figure}
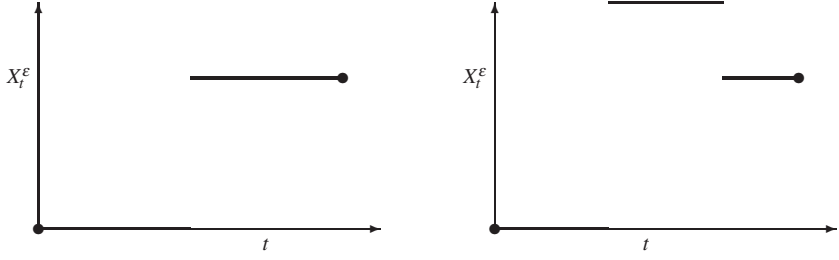

Let $N$ be the number of jumps. If $N\!\leq\! 1$, then $Z\!=\!0$, while if $%
N\!=\!2$, then $Z\leq \min \left( \left\vert Y_{1}\right\vert ,\left\vert
Y_{2}\right\vert \right)$. This can be seen from the behavior of $%
X_t^\varepsilon$ in the different scenarios illustrated in Figure \ref%
{fig:jump}. More generally, if $N\!=\!k,\ k\!\geq\! 2$, then
\begin{eqnarray*}
Z>x &\Longrightarrow &\exists\ 1\leq j\leq k\!-\!1\text{ s.t. }\left\vert
\sum_{l=1}^{j}Y_{l}\right\vert >x,\ \left\vert
\sum_{l=j+1}^{k}Y_{l}\right\vert >x \\
&\Longrightarrow &\exists\ j_{1},j_{2}\text{ s.t. } \left\vert Y_{j_{1}}
\right\vert >\frac{x}{k\!-\!1},\ \ \left\vert Y_{j_{2}} \right\vert >\frac{x%
}{k\!-\!1}.
\end{eqnarray*}
Since
\begin{eqnarray*}
\lefteqn{ \hspace*{-0.5in} \mathbb{P}\big[ \exists j_{1},j_{2}\text{ s.t. }
\left\vert Y_{j_{1}}\right\vert >\frac{x}{k-1},\ \left\vert
Y_{j_{2}}\right\vert >\frac{x}{k-1}\big] } \\
&\leq &\sum_{\left( j_{1},j_{2}\right) }\mathbb{P}\big[\left\vert
Y_{j_{1}}\right\vert >\frac{x}{k-1},\ \left\vert Y_{j_{2}}\right\vert >\frac{%
x}{k-1}\big] \\
&=&\frac{k\left( k-1\right) }{2}\mathbb{P}\big[\left\vert Y_{1}\right\vert >%
\frac{x}{k-1}\big]^{2}.
\end{eqnarray*}
it follows that
\begin{eqnarray}
\mathbb{E}\left[ Z^{p}\mid N\!=\!k \,\right] &=& p \int x^{p-1}\mathbb{P}\left[Z>x\mid N\!=\!k\right] \D x  \notag \\
&\leq &\frac{k\,(k\!-\!1)}{2}\, p\int x^{p-1}\mathbb{P}\big[\left\vert
Y_{1}\right\vert >\frac{x}{k-1}\big]^{2}\D x  \notag \\
&=&\frac{k\,(k\!-\!1)}{2}\,\frac{p}{\lambda_\varepsilon^2}\int x^{p-1}\big(
\int_{\{ |z| >x/\left( k-1\right) \}}1_{\left\{ \varepsilon
<\left\vert z\right\vert \right\} }\nu (\D z)\big) ^{2}\D x  \notag \\
&=&\frac{k\,(k\!-\!1)^{p+1}}{2}\, \frac{p}{\lambda_\varepsilon^2} \int
x^{p-1}\big( \int_{\{ |z| >x\}}1_{\left\{ \varepsilon
<\left\vert z\right\vert \right\} }\nu (\D z)\big) ^{2}\D x  \notag \\
&\equiv &d_{k,p}\big( \varepsilon ^{p}+\frac{L_{\varepsilon }\left(
p\right) }{\lambda _{\varepsilon }^{2}}\big) ,  \label{estimate_Z_k}
\end{eqnarray}
where $d_{k,p}=\frac{1}{2} k\,(k\!-\!1)^{p+1}$. We then have
\begin{eqnarray}
\mathbb{E}\left[ Z^{p}\right] &=& \sum_{k=2}^{\infty }\mathbb{E}\left[%
Z^p\mid N\!=\!k\right]\ \mathbb{P}\left[ N\!=\!k\right]  \notag \\
&\leq & \big(\varepsilon^p+\frac{L_\varepsilon(p)}{\lambda_\varepsilon^2}%
\big) \ \exp \big( -\frac{\lambda_\varepsilon}{n}\big)\
\sum_{k=2}^{\infty }d_{k,p}\big( \frac{\lambda _{\varepsilon }}{n}\big)^k
\frac{1}{k!}  \notag
\end{eqnarray}
For $k_{p}=\lceil p\rceil \!+\!2$ there exists $C_{p}$ such that for any $%
k\geq k_{p}$, $d_{k,p}\leq C_{p}\frac{k!}{\left( k-k_{p}\right) !}$, so
\begin{eqnarray*}
\sum_{k=2}^{\infty }d_{k,p}\big( \frac{\lambda _{\varepsilon }}{n}\big)^k\frac{1}{k!} &\leq & \sum_{k=2}^{k_{p}-1}d_{k,p}\big( \frac{%
\lambda_\varepsilon}{n}\big)^{k}\frac{1}{k!} + C_{p}\sum_{k=k_{p}}^{\infty
}\big( \frac{\lambda _{\varepsilon }}{n}\big)^k \frac{1}{%
\left(k-k_{p}\right)!} \\
&\leq & \sum_{k=2}^{k_{p}-1}d_{k,p}\big( \frac{\lambda_\varepsilon}{n}%
\big)^{k}\frac{1}{k!} + C_{p} \big( \frac{\lambda _{\varepsilon }}{n}\big)^k_p \exp\left( \frac{\lambda_\varepsilon}{n}\right) \\
&\leq &K_{p}\big( \frac{\lambda _{\varepsilon }}{n}\big) ^{2}
\end{eqnarray*}
for some constant $K_p$, where the last step uses the fact that $\lambda
_{\varepsilon }\!\leq\! n.$

Therefore, we obtain the final result that
\begin{equation*}
\mathbb{E}\left[ Z^{p}\right] \leq K_{p}\big( \varepsilon ^{p}+\frac{L_{\varepsilon }\left( p\right) }{\lambda
_{\varepsilon }^{2}}\big) \frac{\lambda _{\varepsilon }^{2}}{n^{2}}.
\end{equation*}
\end{proof}


\subsection{Bounding moments of {$\sup_{[0,T]}\left\vert R_{t}^{\protect%
\varepsilon }\right\vert$}}

\begin{proposition}
\label{prop:sig_p}Let $X$ be a scalar L\'{e}vy process with a triple $%
(m,0,\nu )$ and let $R_t^\varepsilon$, $\mu_\varepsilon$, $%
\lambda_\varepsilon$ and $\sigma_\varepsilon$ be as defined in section \ref%
{sec:decomposition}. Then $R_{t}^{\varepsilon }$ satisfies
\begin{equation}
\mathbb{E}\left[ \sup_{[0,T]} \left\vert R_t^\varepsilon\right\vert^p\right]
\leq \left\{
\begin{array}{ll}
K_p\left( T^{p/2}\sigma_\varepsilon^p + T\int_{\{ |z|\leq
\varepsilon \}}\left\vert z\right\vert ^{p}\nu (\D x)\right), & p>2; \\[0.1in]
K_p\, T^{p/2}\sigma_\varepsilon^p, & 1\leq p\leq 2,%
\end{array}%
\right.  \label{estimate:Rp}
\end{equation}%
where $K_{p}$ is a constant depending on $p$.
\end{proposition}

\begin{proof}
For any $1\!\leq\!p\!\leq\!2$, by Jensen's inequality and the Doob
inequality (c.f. Theorem 19 and 20 in \cite{protter04}),
\begin{eqnarray*}
\mathbb{E}\left[ \sup_{0\leq t\leq T}\left\vert R_{t}^{\varepsilon
}\right\vert ^{p}\right] &\leq &\mathbb{E}\left[ \sup_{0\leq t\leq
T}\left\vert R_{t}^{\varepsilon }\right\vert ^{2}\right]^{p/2} \\
&\leq &2^p\, \mathbb{E}\left[ \left\vert R_T^\varepsilon\right\vert ^{2}%
\right]^{p/2} \\
&=&2^p\, T^{p/2}\sigma _{\varepsilon }^{p}.
\end{eqnarray*}

For any $p\!>\!2$, the Burkholder--Davis--Gundy inequality (c.f. Theorem 73 in \cite%
{protter04}) gives
\begin{equation}
\mathbb{E}\left[ \sup_{0\leq t\leq 1}\left\vert R_{t}^{\varepsilon
}\right\vert ^{p}\right] \leq \mathbb{E}\left[ \left[ R^{\varepsilon }\right]
_{1}^{p/2}\right]  \notag
\end{equation}%
where $\left[ R^\varepsilon\right]_t$ is the quadratic variation of $%
R_t^\varepsilon$. We can use the method in the proof of Theorem 66
in\cite{protter04} to get
\begin{eqnarray*}
\mathbb{E}\left[ \left[ R^{\varepsilon }\right] _{1}^{p/2}\right] &\leq
&K_{p}\left[ \left( \int_{\{ |z| \leq \varepsilon \}}z^{2}\nu
(\D z)\right) ^{p/2}+\int_{\{ |z| \leq \varepsilon
\}}\left\vert z\right\vert ^{p}\nu (\D z)\right] \\
&=&K_{p}\left( \sigma _{\varepsilon }^{p}+\int_{\{ |z| \leq
\varepsilon \}}\left\vert z\right\vert ^{p}\nu (\D z)\right)
\end{eqnarray*}%
where $K_{p}$ is a constant depending on $p$.

To extend this result to an arbitrary time interval $[0,T]$ we use a change
of time coordinate, $t^{\prime }\!=\!t/T$, with associated changed L\'{e}vy
measure $\nu^{\prime }(\D z)\!=\!T\,\nu (\D z)$ to obtain
\begin{equation*}
\mathbb{E}\left[ \sup_{[0,T]}\left\vert R_{t}^{\varepsilon }\right\vert ^{p}%
\right] \leq K_{p}\left[ T^{p/2} \sigma _{\varepsilon }^{p/2}
+T\int_{\{ |z| \leq \varepsilon \}}\left\vert z\right\vert
^{p}\nu (\D z)\right] .
\end{equation*}
\end{proof}


\subsection{Bounding moments of $\max_{0\leq i<n}S_n^{(i)}$}


\begin{proposition}
\label{prop:Sn}Let $X$ be a scalar pure jump L\'{e}vy process, with L\'{e}%
vy measure $\nu (x)$ which satisfies
\begin{equation*}
C_{2}\left\vert x\right\vert ^{-1-\alpha }\leq \nu (x)\leq C_{1}\left\vert
x\right\vert ^{-1-\alpha },\ \text{as }\left\vert x\right\vert \leq 1;
\end{equation*}%
for constants $C_{1},C_{2}>0$ and $0\!\leq\! \alpha \!<\!2$. If $S_{n}^{(i)}$
is as defined in section \ref{sec:decomposition}, and $\lambda_\varepsilon%
\leq n$, then for $p\!\geq\!1,$ and arbitrary $\delta \!>\! 0$ there exists
a constant $C_{p,\delta}$, which does not depend on $n, \varepsilon$ such
that
\begin{equation*}
\mathbb{E}\left[ \left( \max_{0\leq i <n}S_{n}^{\left( i\right) }\right) ^{p}%
\right] \leq C_{p,\delta}\ \varepsilon ^{p-\delta }.
\end{equation*}
In the particular case of $\alpha \!=\!0$, such a bound holds with $%
\delta\!=\!0$.
\end{proposition}

\begin{proof}
By Proposition \ref{prop:sig_p}, for $q>2$,
\begin{equation*}
\mathbb{E}\left[ \left( \max_{0\leq i <n}S_{n}^{\left( i\right) }\right)^q%
\right] \ \leq \ n\ \mathbb{E}\left[ \sup_{[0,\frac{1}{n}]}\left\vert
R_{t}^{\varepsilon }\right\vert^q\right] \ \leq \ K_q\left(
n^{1-q/2}\sigma_\varepsilon^q + \int_{\{ |z| \leq
\varepsilon \}}\left\vert z\right\vert^q\nu (\D x)\right).
\end{equation*}
Recalling the definition of $\sigma _{\varepsilon }$ (\ref{sig_ep}), due to
the assumption on $\nu(x)$ we have
\begin{equation*}
\sigma _{\varepsilon }^{q} \ \leq\ \left( \frac{2C_{1}}{2-\alpha }%
\right)^{q/2}\varepsilon^{q-q\alpha /2}, ~~~ \int_{\{ |z|
\leq \varepsilon \}}\left\vert z\right\vert ^{q}\nu (\D x) \ \leq\ \frac{2C_{1}%
}{q-\alpha }\ \varepsilon ^{q-\alpha }.
\end{equation*}

Given $p\geq 1$, for any $q>\max \left( 2,p\right) $, Jensen's inequality
gives us
\begin{eqnarray*}
\mathbb{E}\left[ \left( \max_{0\leq i <n}S_{n}^{\left( i\right) }\right) ^{p}%
\right] &\leq& \mathbb{E}\left[ \left( \max_{0\leq i <n}S_{n}^{(i)}\right)
^{q}\right]^{p/q} \\
& \leq & K_q^{p/q}\left[ \left( \frac{2C_{1}}{2-\alpha }\right)^{q/2} \left(%
\frac{\varepsilon^{-\alpha}}{n}\right)^{q/2-1} \!+\ \frac{2C_{1}}{q-\alpha }%
\right]^{p/q} \varepsilon^{p-\alpha p/q}.
\end{eqnarray*}
If $\alpha\!=\!0$, then the desired bound is obtained immediately. On the
other hand, if $0\!<\!\alpha \!<\!2$, then
\begin{equation*}
\lambda _{\varepsilon } \ \geq\ C_2\int_{\{\varepsilon <\left\vert
z\right\vert <1\}} \frac{1}{\left\vert z\right\vert ^{\alpha +1}}\ \D z \ =\
\frac{2\,C_2}{\alpha }\left( \varepsilon ^{-\alpha }-1\right) .
\end{equation*}
Since $\lambda _{\varepsilon }\!\leq\! n$ it implies that $%
\varepsilon^{-\alpha}\leq \frac{K\alpha }{2\,C_2}n+1, $ and thus $%
\varepsilon^{-\alpha}/n$ is bounded. Hence there exists a constant $C$ such
that
\begin{equation*}
\mathbb{E}\left[ \left( \max_{0\leq i <n}S_{n}^{(i)}\right) ^{p}\right] \leq
C\ \varepsilon ^{p-\alpha p/q},
\end{equation*}
and by choosing $q$ large enough so that $\alpha p/q \leq \delta$ we obtain
the desired bound.
\end{proof}

\subsection{Proof of Theorem \protect\ref{prop:Un}}
\label{sec:proofs-4.2}

\begin{proof}
Provided $\lambda_\varepsilon\!\leq\! n$, by (\ref{decomp:Vp}) and (\ref%
{estimate_Z}) we have
\begin{equation}
\mathbb{E}\left[ D_{n}^{p}\right] \prec \underbrace{\mathbb{E}\left[ \left(
\max_{0\leq i<n}S_{n}^{\left( i\right) }\right) ^{p}\right] }_{1)} +
\underbrace{\varepsilon ^{p}\frac{\lambda _{\varepsilon }^{2}}{n}}_{2)} +
\underbrace{\frac{L_{\varepsilon }\left( p\right) }{n}}_{3)} + \underbrace{%
\left( \frac{\left\vert \mu _{\varepsilon }\right\vert }{n}\right) ^{p}}%
_{4)},  \label{estimate:Vp_general}
\end{equation}
where the notation $u\prec v$ means there exists constant $c>0$ independent
of $n$ such that $u<cv$.

We can now bound each term, given the specification of the L\'{e}vy measure,
and if we can choose appropriately how $\varepsilon\rightarrow 0$ as $%
n\rightarrow \infty$ so that the RHS of (\ref{estimate:Vp_general}) is
convergent, then the convergence rate of $\mathbb{E} \left[ D_{n}^{p}\right]
$ can be bounded.

For $0\!<\! x \!<\!1$,%
\begin{eqnarray}
\lambda_x &\leq & C_1 \int_{x <\left\vert z\right\vert <1} \frac{1}{%
\left\vert z\right\vert^{\alpha+1}}\ \D z +\int_{1<\left\vert z\right\vert }
\exp \left( -C_{3}\left\vert z\right\vert\right)\, \D z  \notag \\[0.1in]
&\leq &\left\{
\begin{array}{ll}
2C_1\log \frac{1}{x}+l_1, & \alpha =0; \\[0.1in]
l_2\, x^{-\alpha }, & 0<\alpha <2.%
\end{array}
\right.  \label{estimate:lam1}
\end{eqnarray}
where $l_1, l_2$ are constants with $l_2\geq 2\, C_3^{-1}$, while for $x
\!\geq \! 1$,
\begin{equation*}
\lambda_x\ \leq\ \int_{x<\left\vert z\right\vert } \exp \left(-C_3\left\vert
z\right\vert\right)\, \D z \ =\ 2\, C_3^{-1} \exp( -C_3\, x).
\end{equation*}
If $\alpha >0$, then
\begin{eqnarray}
L_{\varepsilon }(p) &=& p\int_{x>\varepsilon }x^{p-1}\lambda_{x}^{2}\ \D x
\notag \\
&\leq &l_2^2\ p\int_{x>\varepsilon }x^{p-1}\left(\ 1_{\left\{ x<1\right\}
}x^{-2\alpha }+1_{\left\{ x>1\right\} }\exp \left(-2 C_3 x \right)\ \right) %
\D x  \notag \\
&\leq &\left\{
\begin{array}{ll}
l_{4}, & p>2\alpha ; \\[0.05in]
l_{4}\log \frac{1}{\varepsilon }+l_{5}, & p=2\alpha ; \\[0.05in]
l_{4}\varepsilon ^{-2\alpha +p}+l_{5}, & p<2\alpha .%
\end{array}%
\right.  \label{estimate:L}
\end{eqnarray}%
where $l_3, l_4, l_5$ are additional constants. If $\alpha\!=\!0$, it is
easily verified that $L_{\varepsilon }(p)$ is bounded for $p\!\geq\!1$, so (%
\ref{estimate:L}) applies equally to this case.

Given $0\!<\!\varepsilon \!<\!1$ we have%
\begin{eqnarray}
\left\vert \mu _{\varepsilon }\right\vert &=&\left\vert m-\int_{\varepsilon
<\left\vert z\right\vert <1}z\ \nu (\D z)\right\vert  \notag \\
&\leq &\left\{
\begin{array}{ll}
\left\vert m\right\vert + \left\vert C_1\!-\!C_2 \right\vert \frac{%
\varepsilon^{1-\alpha }-1}{\alpha -1}, & \alpha \neq 1; \\[0.1in]
\left\vert m\right\vert +\left\vert C_1\!-\!C_2\right\vert \log \frac{1}{%
\varepsilon }, & \alpha =1.%
\end{array}%
\right.  \label{estimate:mu}
\end{eqnarray}


Subject to the condition that $\lambda_\varepsilon \leq n$, we now have

\begin{enumerate}[1)]
\item By Proposition \ref{prop:Sn},
\begin{equation*}
\mathbb{E}\left[ \left( \max_{i=1,\ldots ,n}S_{n}^{\left( i\right) }\right)
^{p}\right] \prec \varepsilon ^{p-\delta }, ~~ \mbox{for any } \delta\!>\!0.
\end{equation*}

\item By (\ref{estimate:lam1}),
\begin{equation*}
\varepsilon ^{p}\frac{\lambda_\varepsilon^2}{n}\ \prec\ n^{-1} \times
\left\{
\begin{array}{ll}
\varepsilon^{p}\log \frac{1}{\varepsilon}, & \alpha =0; \\[0.1in]
\varepsilon^{p-2\alpha }, & 0<\alpha <2.%
\end{array}%
\right.
\end{equation*}

\item By (\ref{estimate:L}),
\begin{equation*}
\frac{L_{\varepsilon }(p)}{n}\ \prec\ n^{-1}\times \left\{
\begin{array}{ll}
1, & p>2\alpha ; \\[0.05in]
\log \frac{1}{\varepsilon }, & p=2\alpha ; \\[0.05in]
\varepsilon ^{-2\alpha +p}, & p<2\alpha .%
\end{array}
\right.
\end{equation*}

\item By (\ref{estimate:mu}),
\begin{equation*}
\left( \frac{\left\vert \mu _{\varepsilon }\right\vert }{n}\right) ^{p}\prec
n^{-p}\times \left\{
\begin{array}{ll}
1+\left\vert C_{1}-C_{2}\right\vert ^{p}\varepsilon ^{p\left( 1-\alpha
\right) }, & \alpha >1; \\[0.05in]
1+\left( \left\vert C_{1}-C_{2}\right\vert \log \frac{1}{\varepsilon }%
\right)^p, & \alpha =1; \\[0.05in]
1, & \alpha <1.%
\end{array}%
\right.
\end{equation*}
\end{enumerate}

In the following we assume $C_{1}\neq C_{2}.$

\begin{enumerate}
\item $p\geq 2\alpha $.

If we choose $\varepsilon = C\, n^{-2/p}$, then $\lambda_\varepsilon \prec
\varepsilon^{-\alpha} \prec n^{2\alpha/p}, $ and the constant $C$ can be
taken to be sufficiently small so that $\lambda_\varepsilon\leq n$ for
sufficiently large $n$.

Taking $\delta < p/2$, we find that the dominant contribution to (\ref%
{estimate:Vp_general}) comes from 3), giving the desired result that
\begin{equation*}
\mathbb{E}\left[ D_{n}^{p}\right] \prec \left\{
\begin{array}{ll}
n^{-1}, & p > 2 \alpha; \\[0.05in]
\log n / n, & p = 2 \alpha.%
\end{array}%
\right.
\end{equation*}

\item $1\leq p<2\alpha$.

We can use H\"{o}lder's inequality to give $\mathbb{E}\left[ D_n^p \right] \
\leq\ \mathbb{E}\left[ D_n^{2\alpha}\right]^{\frac{p}{2\alpha}} \ \prec \
\left( \log n / n\right)^{\frac{p}{2\alpha }}. $
\end{enumerate}

For a spectrally negative process, $\sup_{[0,\frac{1}{n}]}X_t^\varepsilon -
\left(X_{\frac{1}{n}}^\varepsilon\right)^+ = 0$, since $X_{t}$ doesn't have
positive jumps, and hence
\begin{equation*}
\mathbb{E}\left[ D_{n}^{p}\right] \leq \mathbb{E}\left[ \left(\max_{0\leq
i<n} S_n^{(i)}\right)^p\right] + \left(\frac{\left\vert
\mu_\varepsilon\right\vert }{n}\right)^p.
\end{equation*}

We can take $\varepsilon = C\ n^{-1/\alpha}$ with the constant $C$ again
chosen so that $\lambda_\varepsilon \leq n$ for sufficiently large $n$. We
then obtain
\begin{equation*}
\mathbb{E}\left[ D_{n}^{p}\right] \prec \left\{
\begin{array}{ll}
n^{-p/\alpha + \delta}, & \alpha \geq 1; \\
n^{-p}, & \alpha<1.%
\end{array}%
\right.
\end{equation*}
for any $\delta\!>\!0$.
\end{proof}

\subsection{Proof of Proposition \protect\ref{prop:VG-barrier}}
\label{sec:proofs-5.4}

We decompose the term we want to bound into parts and then balance their asymptotic orders to get desired result. 

Note that
$\mathbf{1}_{\{\widehat{m}_{n}<B\}}-\mathbf{1}_{\left\{ m<B\right\}} =1$
only if either $m$ is close to the barrier or the difference between
discretely and continuously monitored maximum $D_n=m-\widehat{m}_n$ is
large. More precisely,
\begin{equation*}
\left\{ \mathbf{1}_{\{\widehat{m}_n<B\}} - \mathbf{1}_{\left\{m<B\right\}
}=1\right\} \subset F\cup G,
\end{equation*}
where $F := \left\{ B\leq m \leq B+n^{-r}\right\}$ and $G := \left\{
D_{n}>n^{-r}\right\}$ 
for an $r>0$ to be determined. Hence
\begin{equation*}
\mathbb{E}\left[ \mathbf{1}_{\{\widehat{m}_n<B\}} - \mathbf{1}%
_{\left\{m<B\right\} }\right] \leq \mathbb{P}\left[ F\right] +\mathbb{P}\left[G\right] .
\end{equation*}

Due to the locally bounded density for $m$, $\mathbb{P}\left[ F\right] =\bO\left(
n^{-r}\right)$.

If we denote
\begin{equation*}
Z_n^{(i)}=\sup_{[0,\frac{1}{n}]}\Delta^{(i)}X_{t}^{\varepsilon } - \left(
\Delta^{(i)} X_{\frac{1}{n}}^{\varepsilon } \right)^{+}.
\end{equation*}
where $\Delta^{(i)}X_{t}$ is as defined previously in Section \ref%
{sec:decomposition}, then (\ref{estimate:Vp_raw}) gives
\begin{equation*}
D_{n}\leq \max_{0\leq i < n} Z_n^{(i)} + \frac{\left\vert
\mu_\varepsilon\right\vert }{n} +\max_{0\leq i < n} S_n^{(i)}
\end{equation*}
For $\alpha<1$, $\mu_\varepsilon$ is bounded, so $\left\vert
\mu_\varepsilon\right\vert \leq {\ \textstyle \frac{1}{2} }n^{1-r}$, for
sufficiently large $n$. Hence,
\begin{eqnarray*}
\mathbb{P}\left[ D_{n}>n^{-r}\right] & \leq & \mathbb{P}\left[ \max_{0\leq i
< n} Z_n^{(i)} + \max_{0\leq i < n} S_n^{(i)} > {\ \textstyle \frac{1}{2} }%
n^{-r}\right] \\
& \leq & \mathbb{P}\left[ \max_{0\leq i < n} Z_n^{(i)} >{\ \textstyle \frac{1%
}{4} } n^{-r}\right] + \mathbb{P}\left[ \max_{0\leq i < n} S_n^{(i)} >{\ %
\textstyle \frac{1}{4} } n^{-r}\right] .
\end{eqnarray*}

Now, $\max_{0\leq i < n}Z_n^{(i)} > 0 $ requires that there are at least two
jumps in one of the $n$ intervals. The probability of two jumps in one
particular interval is 
$$1 - \exp\left(-\, {\ \textstyle \frac{%
\lambda_\varepsilon}{n} }\right) \left(1 + {\ \textstyle \frac{%
\lambda_\varepsilon}{n} }\right) \prec \left( {\ \textstyle \frac{%
\lambda_\varepsilon}{n} } \right)^2 $$ if $\lambda_\varepsilon \leq n$, and
hence
\begin{equation*}
\mathbb{P}\left[ \max_{0\leq i < n} Z_n^{(i)} >{\ \textstyle \frac{1}{4} }
n^{-r}\right] \prec \frac{\lambda_\varepsilon^2}{n}.
\end{equation*}

We use the Markov inequality for the remaining term. 
According to Proposition \ref{prop:sig_p}, $\mathbb{E}\left[ \max_{0\leq i <
n}\left( S_{n}^{\left( i\right) }\right)^{p}\right] \prec \varepsilon
^{p-\delta } $ and so
\begin{equation*}
\mathbb{P}\left[ \max_{0\leq i < n}S_{n}^{\left( i\right) }>{\ \textstyle
\frac{1}{4} } n^{-r}\right] \ \prec \ \mathbb{E}\left[ \max_{0\leq i <
n}\left(S_n^{(i)}\right)^{p}\right]\ /\ \left({\ \textstyle \frac{1}{4} }
n^{-r}\right)^p \ \prec \ \varepsilon ^{p-\delta }n^{rp}.
\end{equation*}

Combining these elements, provided $\lambda_\varepsilon\leq n$, we have
\begin{equation*}
\mathbb{E}\left[ \mathbf{1}_{\{\widehat{m}_n<B\}}-\mathbf{1}%
_{\left\{m<B\right\} }\right] \prec n^{-r}+\varepsilon ^{p-\delta }n^{rp} +
\frac{\lambda_\varepsilon^2}{n}.
\end{equation*}

Equating the first two terms on the right hand side gives $\varepsilon =
n^{-r(1+p)/(p-\delta)}$.

\vspace{0.1in}

If $\alpha\!=\!0$, then $\lambda_\varepsilon \prec \log \frac{1}{\varepsilon}
\prec \log n$, so $\lambda_\varepsilon=\bo(n) $ is satisfied. We also have $%
\frac{\lambda_\varepsilon^2}{n} \prec \frac{\left(\log n\right)^2}{n}, $ and
therefore for any $r\!<\!1$ we have $\mathbb{E}\left[ \mathbf{1}_{\{\widehat{%
m}_n<B\}} -\mathbf{1}_{\left\{m<B\right\} }\right] \prec n^{-r}. $

\vspace{0.1in}

If $0\!<\!\alpha \!<\!2$, then $\lambda_\varepsilon \prec
\varepsilon^{-\alpha } \prec n^{r\alpha (1+p)/(p-\delta)}$, and hence $\frac{%
\lambda_\varepsilon^2}{n} \prec n^{-1+2r\alpha (1+p)/(p-\delta)}$. Balancing
$n^{-r}$ and $n^{-1+2r\alpha (1+p)/(p-\delta) }$, gives $\lambda_\varepsilon=%
\bo(n)$ and
\begin{equation}
r=\left( 1 + 2\alpha\ \frac{1\!+\!p}{p-\delta}\right)^{-1}.  \label{eq:r}
\end{equation}

Since $r\rightarrow \frac{1}{1+2 \alpha}$ as $\delta\rightarrow 0$, and $%
p\rightarrow \infty$, for any fixed value of $r<\frac{1}{1+2 \alpha}$ it is
possible to choose appropriate values of $p$ and $\delta$ to satisfy (\ref%
{eq:r}) and thereby conclude that $\mathbb{E}\left[ \mathbf{1}_{\{\widehat{m}%
_n<B\}} -\mathbf{1}_{\left\{m<B\right\} }\right] \prec n^{-r}. $

\begin{acknowledgements}
This work was supported by the China Scholarship Council and the 
Oxford-Man Institute of Quantitative Finance. We would like to thank 
Ben Hambly, Andreas Kyprianou, Loic Chaumont, Jacek Malecki and Jose Blanchet  for their helpful comments.
\end{acknowledgements}

\end{document}